\documentclass[sigconf]{acmart}

\usepackage{graphicx}
\graphicspath{ {images/} }
\usepackage{booktabs}
\usepackage[export]{adjustbox}		

\usepackage{algorithm}
\usepackage[noend]{algpseudocode}
\usepackage{mdwmath}
\usepackage{mdwtab}
\usepackage[utf8]{inputenc}
\usepackage[english]{babel}
\usepackage{breqn}

\usepackage{pgfplots}
\usepackage{tikz}

\usepackage{amsmath,amssymb}
\usepackage{amsthm}
\usepackage{xspace}
\usepackage[T1]{fontenc}

\usepackage{graphicx}
\usepackage[font=small,skip=4pt]{caption}
\usepackage[font=footnotesize,skip=0pt]{subcaption}

\usepackage{color}
\usepackage{colortbl}
\usepackage{bbm}
\usepackage{url}
\usepackage{breakurl}

\usepackage[english]{babel} % English language/hyphenation

\usepackage{listings}
\usepackage{xcolor}
\usepackage{multirow}
\usepackage{tabularx}
\usepackage{bbold}
\usepackage{enumitem}
\setlist{leftmargin=5.5mm}
\usetikzlibrary{matrix}
\usepgfplotslibrary{groupplots}
\pgfplotsset{compat=newest}

\newcommand{\dante}{{\sc Airmed}\xspace}
\newcommand{\omnet}{{OMNeT++}\xspace}
\newcommand{\opt}{${\mathcal{O}}$\xspace}
\newcommand{\opti}{{\mathcal{O}}}
\newcommand{\adv}{${\mathcal{A}}$\xspace}

\newcommand{\blank}{{blank}}

\newcommand{\whp}{{$w.h.p$}}

\newcommand{\nodes}{{\bf N}}

\newcommand{\sram}{${\sf SecRAM}$}
\newcommand{\rom}{${\sf ROM}$}
\newcommand{\code}{${\sf code}$}
\newcommand{\data}{${\sf data}$}
\newcommand{\rf}{${\sf NIC}$}
\newcommand{\msg}{${\sf MSG}$}

\newcommand{\req}{{\sf req}}
\newcommand{\resp}{{\sf resp}}
\newcommand{\ttl}{{\sf ttl}}

\newcommand{\ver}{${\sf ver}$}
\newcommand{\cert}{${\sf cert}$}

\newcommand{\posi}{{\sf pos}}
\newcommand{\dati}{{\sf data}}

\newcommand{\adhoc}{{Mesh}}

\newcommand{\selfcheck}{${\sf selfcheck}$}
\newcommand{\attest}{${\sf attest}$}
\newcommand{\sign}{${\sf sign}$}
\newcommand{\rectify}{${\sf rectify}$}

\newcommand{\extl}{{\sf ext}}
\newcommand{\intl}{{\sf int}}

\newcommand{\vinay}[1]{\textcolor{blue}{{\bf Vinay:} #1}}

\makeatletter
\def\thickhline{%
  \noalign{\ifnum0=`}\fi\hrule \@height \thickarrayrulewidth \futurelet
   \reserved@a\@xthickhline}
\def\@xthickhline{\ifx\reserved@a\thickhline
               \vskip\doublerulesep
               \vskip-\thickarrayrulewidth
             \fi
      \ifnum0=`{\fi}}
\makeatother

\newlength{\thickarrayrulewidth}
\pgfplotsset{
/pgfplots/lineLegend/.style={
legend image code/.code={
\draw[black](-0.05cm,0cm) -- (0.3cm,0cm);%
   }
  }
}

\begin{document}

\author{Sourav Das}
\authornote{Part of the work was done when the author was at IIT Bombay}
\affiliation{University of Illinois at Urbana-Champaign}
\email{souravd2@illinois.edu}

\author{Samuel Wedaj}
\affiliation{Department of Computer Science and Engineering., Indian Institute of Technology Delhi}
\email{samuel.wed@cse.iitd.ac.in}

\author{Kolin Paul}
\affiliation{Department of Computer Science and Engineering., Indian Institute of Technology Delhi}
\email{kolin@cse.iitd.ac.in}

\author{Umesh Bellur}
\affiliation{Department of Computer Science and Engineering, Indian Institute of Technology Bombay}
\email{umesh@cse.iitb.ac.in}

\author{Vinay Joseph Ribeiro}
\affiliation{Department of Computer Science and Engineering, Indian Institute of Technology Bombay}
\email{vinayr@iitb.ac.in}

\title{\dante: Efficient Self-Healing Network of Low-End Devices}

\begin{abstract}
The proliferation of application specific cyber-physical systems coupled with the emergence of a variety of attacks on such systems (malware such as Mirai and Hajime) underlines  the need to secure such networks. Most existing security efforts have focused on only detection of the presence of malware. However given the ability of most attacks to spread through the network once they infect a few devices, it is important to contain the spread of a virus and at the same time systematically cleanse the impacted nodes using the communication capabilities of the network. Toward this end, we present \dante\ - a method and system to not just detect corruption of the application software on a IoT node, but to self correct itself using its neighbors. \dante's decentralized mechanisms prevent the spread of self-propagating malware and can also be used as a technique for updating application code on such IoT devices. 
Among the novelties of \dante\ are  a novel bloom-filter technique along with hardware support to identify position of the malware program
from the benign application code,
an adaptive self-check for computational efficiency, and a uniform random-backoff and stream signatures for secure and bandwidth
efficient code exchange to correct corrupted devices.  We assess the performance of \dante,  using the embedded systems security architecture of TrustLite in the \omnet\ simulator. The results show that \dante
scales up to thousands of devices, ensures guaranteed update of the entire network, and can recover 95\% of the nodes in ~10 minutes
in both internal and external propagation models. Moreover, we evaluate memory and communication costs and show that  \dante is efficient and incurs very low overhead.

\end{abstract}


\maketitle

% !TEX root = ../main.tex
\section{Introduction}
\label{sec:intro}
%
\begin{comment}
\vinay{Intro is too long and unclear. We need to get quicker to the problem we are solving and our solution. How about this flow: (1) Malware in low-end IoT is an important issue. (2) So far, work on malware detection only in such network. (3) Our problem: not just detection, but curing and updating (need to be clear of difference between curing and updating) a network of low-end devices. (4) Why problem is non-trivial, challenges due to low computation, and possibly no Internet connectivity (5) Brief overview of solution and how we address challenges. }
\end{comment}

Application-specific low-end devices have become ubiquitous 
in safety-critical systems such as hazard control, airplanes,
nuclear reactors, etc. A recent Gartner report estimates that there will 
be more than 20 billion Internet-of-Things~(IoT) devices by 
the end of the year 2020~\cite{Gartner}.\footnote{We use the 
terms ``device'' and ``node'' interchangeably in the paper.}
As the use of such devices becomes imperative in mission critical systems, their security is of immense 
concern~\cite{kumar2019all}. 
Attacks on a nuclear power plant using Stuxnet~\cite{trautman2017industrial}, 
large scale Distributed-Denial-of-Service~(DDoS) attacks using 
IoT Botnets such as Mirai and Hajime~\cite{antonakakis2017understanding,kolias2017ddos}, potential disruption of power-grids using 
high wattage devices~~\cite{soltan2018blackiot}, and malware which can rapidly spread citywide 
using deployed Phillips hue bulbs, illustrate 
the importance of ensuring that such networks are secure and/or can recover quickly and cheaply from attacks~\cite{ronen2017iot}.
% \umesh{Based on the CCS review comment "1. the assumptions made in this paper are not very consistent with the low-end IoT devices infected by Mirai in the real world. ", do we even want to refer to these here? Or have we changed assumptions? I suggest we hypothesize the attacks and not compare with actual attacks that have happened.}
% \sourav{Existing works in the same area also considers similar kinds of 
% assumption. As I mention later in the introduction that, we want to study 
% the problem in somewhat abstract sense.}

Existing security research on low-end devices only focuses on detecting the presence of malware using Remote Attestation (RA)~\cite{asokan2015seda,seshadri2004swatt,ambrosin2016sana}
and Machine Learning (ML)~\cite{azmoodeh2018robust,
haddadpajouh2018deep}. RA allows a trusted verifier to detect a compromised device or network of connected devices.
Similarly, the core idea behind the ML-based 
detection of an attack is to train an ML model with
historical network traces and use inferences to detect network 
intrusion in real-time. 
% \vinay{I am curious to know where DANTE fits in with other methods to detect and stop the spread of malware. We can  use ML in addition to Dante. So can we say it is complementary to some other methods?}

Although these approaches are useful, their scope is limited. This
is because, corrupt devices can not only malfunction, they can even 
spread the malware to other nodes.
Furthermore, as we illustrate~(ref.~\S\ref{sec:code update}) an
intelligent adversary can fatally prevent a significant 
fraction of honest nodes from correcting themselves via updates. 
Existing approaches fail to restrict a self-propagating
malware from compromising the entire network~\cite{ronen2017iot,
antonakakis2017understanding,kolias2017ddos,bertino2017botnets,
fernandes2016security}. Also, they do not tackle how 
to recover a compromised device in the presence of a root 
privileged adversary. 
A naive scheme of deploying a vulnerability patch over the 
network to 
update the corrupt device would not work because an adversary with 
access to the incoming 
network messages can simply drop the update messages. 
Finally, such efforts also make strong assumptions 
such as the trusted party can communicate with the network 
at all times; corrupted devices voluntarily communicate with 
a trusted party, and so on that rarely hold in real 
cyber-physical networks. 

% Although these approaches are useful, 
% their scope is limited 
% to malware detection. 
% \vinay{Next is a strong statement. Why cannot an ML algorithm detect which nodes are corrupt and quarantine them, thereby stopping a spread of malware? We have to be careful here.} 
%\umesh{is network privileged adversary a standard term??} 
% Furthermore, as we illustrate~(\S\ref{sec:code update}) an
% intelligent adversary can fatally prevent a significant 
% fraction of honest nodes from correcting themselves via 
% updates. 

\vspace{0.5mm}
\noindent
Motivated by the above, we answer the following questions:
%in this paper: 
\begin{itemize}[noitemsep,topsep=1pt]
    \item How to securely and efficiently detect the presence 
    of a self-propagating malware (including zero-day
    attacks) in a network of heterogeneous low-end devices?
    % and prevent it from spreading to the entire network?
    % \umesh{detecting and self correction once a node has a patch needs hardware support. prevent from spreading is another  part and there is an explicit protocol to warn neighbors. I would like to see this as the second question}
    \item Once malware is detected, how to prevent it from spreading 
    to the entire network and how to securely heal the corrupt devices in 
    a decentralized manner without the intervention from an external 
    trusted party while ensuring minimal overhead?
\end{itemize}

% Numerous examples of attacks with self-propagating malware 
% are known for networks considered in our
% paper~\cite{ronen2017iot,antonakakis2017understanding,
% kolias2017ddos, bertino2017botnets,fernandes2016security}.
% Prominent attack technique among these is the infamous buffer 
% overflow~\cite{cowan1998stackguard,altium2017buffer,
% senrio2017devil} with more than in 500 mentions Common 
% Vulnerabilities and Exposures~(CVE) for the year 2019 
% only. Also, heterogeneity here refers to the difference in 
% hardware and software resources the devices are equipped 
% with.
%
As our solution, we present \dante\footnote{Goddess of healing (in Irish mythology), 
%who is 
known for her
prowess in healing those who fell in battle~\cite{Airmed}.}, the first decentralized 
mechanism to recover a heterogeneous network of low-resource 
cyber-physical systems (CPS) in the presence of self-propagating 
malware. In addition to device correction, \dante\
further assists in the critical issue of over-the-air code updates. 
Specifically, it ensures that all devices in the network get updated.
% \umesh{This is a different problem - one that may share the same solution as we have but I have always said it will confuse readers. Also is there a constraint on speed of update like you mentioned in the speed of recovery earlier?}
% \sourav{It is not clear what do you mean by constraint on speed
% of update? Do you mean how long will it take for the entire 
% network to get updated?}
We would 
like to emphasize that, to ensure that our solution remain 
applicable in a more general sense, we deliberately avoid 
implementation specific details of IoT devices and study 
the problem in abstract sense. 

At its core, every device in \dante\ performs a periodic 
self-check of the application that the device is running. 
\dante\ assumes (readily available) minimal hardware support for 
the self-check~\cite{eldefrawy2012smart,koeberl2014trustlite,
trustile2018impl}. 
% \umesh{Werent there many questions about his in the CCS review?} \sourav{They asked us to implement this atop
% one of these.} 
During the self-check, if a device detects 
that it has been corrupted, it disables the execution of its 
application code. Then the device seeks assistance from its 
neighbors  to recover itself with the correct/updated code.
We refer such a device as {\em \blank} device.
Although, execution of applications are disabled in a \blank\ 
device, we ensure~(ref.~\S\ref{sub:hardware}) that it can still communicate with its peers and run the recovery protocol. 

\vspace{0.5mm}
\noindent
{\bf Challenges.}
% \vinay{you have stated the second and third challenges. But it is not clear what the first challenge is. Is it determining "frequency of self-checks"? Also, is this the most important challenge? We should put the most important challenges first.}
The resource-constrained nature of low-end devices raises a series 
of challenges in designing secure and efficient correcting protocols. 
First challenge is to reduce the trade-off between 
bandwidth usage and delay in the correction. Specifically, if a 
device has $N$ neighbors, the procedure of asking each neighbor 
to transmit the correct application program has a high 
bandwidth cost. Alternatively, asking the neighbors transmit
the correct application in a round-robin manner can lead to 
long delay in correction. Furthermore, these approach also 
introduces security vulnerability, as a malicious neighbor can 
send an incorrect application code to exhaust the bandwidth 
resources of an honest device.

Second challenge arises due to the fact that each self-check 
is expensive and involves interrupting normal execution flow. 
Hence, we want self-checks to be rare, but a rare self-check 
will allow the malware to stay undetected for longer duration 
leading to a faster spread. 

A third challenge is to efficiently identify the modified 
portion of the application code to avoid downloading the entire 
application. This can significantly reduce the bandwidth overhead
in devices with large application code. Naive approaches, such
as storing a hash of chunks of the application code in a secure
memory, increases the size of secure memory. Alternate approach 
of participating in an interactive protocol as in~\cite{ibrahim2019healed}, requires $O(\log z)$ rounds of 
communication in the worst case to identify a single modified 
code chunk from a total of $z$ chunk. 

Additional challenges include efficient authentication of messages 
to prevent replay attacks, identifying appropriate and realistic 
network constraints to ensure that a blank device can communicate 
with its neighbors. We address all of these challenges in 
%the rest of
this paper.

\noindent In summary we make the following contributions:
\begin{itemize}[noitemsep,topsep=1pt]
    \item We present \dante, the first decentralized, secure and 
    a resource-efficient mechanism that ensures recovery of devices 
    in the presence of a self-propagating malware in a heterogeneous network 
    without the intervention of a trusted entity. We also demonstrate 
    that \dante\ mitigates several critical limitations of prevalent 
    secure device update schemes.

    \item We perform a rigorous theoretical analysis of various 
    mechanisms used in \dante and illustrate their efficiency
    over naive schemes. We prove that \dante\ guarantees the 
    recovery and update of all the devices and under specific
    assumptions. 

    % \item We perform a comprehensive simulation of \dante\ atop OMNeT++
    % simulator and perform a rigorous evaluation of it with Mesh, 
    % Binary and Ternary Tree topologies. Our evaluation demonstrates
    % that \dante\ can scales up to thousands of devices, heal 95\%
    % network in 10 minutes, in the presence of both internal and
    % external adversaries, guarantee update of the entire network, 
    % and has low overhead.
    % \umesh{We present a thorough empirical analysis of \dante with multiple topologies over a simulator using OMNET++ to show that it can scale to 1000s of devices blah blah}
    \item We present a thorough empirical analysis of \dante\ with 
    multiple topologies using \omnet\ simulator. Our evaluation 
    illustrates that \dante\ can scale to 1000s of devices, heal 
    95\% network in 10 minutes, and guarantee update of the entire network, while ensuring low overheads. 
\end{itemize}

\noindent
{\bf Organization.}
We present the System Overview, Threat Model and Required
Connectivity in~\S\ref{sec:system}. Details about device initialization and network setup are given in~\S\ref{sec:setup}.
\S\ref{sec:design} presents our detailed design of device
correction followed by details of code update 
in~\S\ref{sec:code update}. We then, theoretically analyze 
our design choices in~\S\ref{sec:analysis}. Simulation and 
Evaluation details are given in~\S\ref{sec:simulation} 
and~\S\ref{sec:evaluation}. A few related works are described 
in~\S\ref{sec:related work}. We finally conclude with a 
discussion in~\S\ref{sec:conclusion}.

\section{System Model}
\label{sec:system}
%
% \vinay{This idea of storing binaries of other devices may not be done in practice. So are we restricting only HR devices to have this property? Or LR also? We have to give some justification.
% If we prove that without this, it is NOT POSSIBLE to correct a network, then our assumption will be reasonable. No other scheme will be able to heal the network without this assumption.}
We consider a connected network of $N$ low-end devices, where 
$i^{\rm th}$ device, $n_i$, runs a set of applications $B_i$. Further, 
we allow devices to store binaries $C_i$ (${\sf .bin}$ files) of other applications 
and transfer these binaries on request from a device connected to 
it. Let $A_i=B_i\cup C_i$. 
All devices and the associated binaries are initialized and deployed
by a trusted third party \opt. 
As devices often have heterogeneous resources such as memory,
bandwidth, and power, based on the resources available at these
devices, we classify them into
{\em Low Resource}~(LR) and {\em High Resource}~(HR) devices.
A LR device only communicates with the devices 
it is directly connected to; such as devices in its wireless
transmission range. Furthermore, LR device only responds to 
a code request of a neighbor if and only if it already possesses 
the requested data. In contrast, HR devices can employ fault-tolerant 
routing algorithms such as Ariadne~\cite{hu2005ariadne}, 
SAR~\cite{yi2001security} to communicate with HR devices through
a sequence of other HR devices. Also, HR device forwards all 
kinds of messages as long as it can validate the signature that 
the message carries. Note that, \dante\ will also work in 
network with only LR or HR devices as long as the connectivity
requirement specified in~\S\ref{sub:connectivity} are satisfied.
% There are two primary differences between a HR and a LR device: 
% {\em Second,} LR devices do not 
% operate as a relay \vinay{What about in the case of limited broadcast? Do LR nodes then relay messages?} to forward \dante\ related message. Also,
% a LR device responds to a code request of a neighbor if and only 
% if it already possesses the requested data. Alternatively, 
% HR device forwards all kinds of messages as long as it can 
% validate the signature that the message carries. 

\subsection{Threat Model}
\label{sub:adversary}
A device is called {\em corrupt} if any of its binaries from $A_i$ 
is modified by an adversary \adv.
% \footnote{One can consider more
% specific definition of corrupted such as modification of binaries 
% from $B_i$ etc. For brevity we consider the simplest definition.}
Similar to existing works, we consider software-only-attacks~\cite{asokan2015seda,eldefrawy2012smart,
ambrosin2016sana}. Hence, at the application layer, a corrupt device 
can arbitrarily deviate from the specified \dante\ protocol. Also, 
\adv\ can drop arbitrary network packets that arrive or leave a 
corrupt device. 
%
% We consider software-only-attacks 
% where the attacker installs a malicious code in the compromised device. % We disregard hardware and Medium Access Control layer attacks
% because such attacks are both difficult of launch\vinay{why are MAC attacks hard to launch?} and defend.\vinay{saying that we do not consider something because it is hard to defend is not a good argument. You give some justification in the Hardware Module section. Maybe part of that argument or a gist of it can be put here. Also, the networking stack is not MAC + Application.  There are other layers too, with application at the top. We can maybe say we neglect attacks at layers "lower than application layer, such as network and MAC", assuming that the firmware for this is in ROM etc. and focus on application layer attacks.}
% For a hardware attack, \adv needs to physically capture a 
% device, which increases the chances of exposing the attacker. 
% Further often sophisticated machineries are required to 
% modify hardware or to extract information from it.
% A nice review on this topic is covered in~\cite{costan2016intel}.
% However, \adv\ can drop arbitrary network packets
% that arrive or leave a corrupt device. 
%
Next, based on the malware propagation model, we classify \adv\ into
two categories: {\em Internal} and {\em External}. 

With an internal adversary, we assume \adv\ has, at time 0, 
corrupted $f$ fraction of devices. Each corrupt device, say 
$n_i$, spreads the malware
%to its neighboring devices
as follows. 
First, $n_i$ chooses 
one of its neighbors at random, waits from a time-interval
drawn from an 
% \vinay{why is exponential a good assumption? Why would it not immediately try to spread? Do real malware deliberately wait for some time, in order to prevent intrusion detection tools which may be monitoring activity?} 
exponential distribution with parameter $\lambda_\intl$
and corrupts the chosen neighbor. 
If the neighbor is already corrupt, its state remain 
unchanged. All corrupt devices independently repeat this process for the 
entire duration they remain corrupt. 
Intuitively this model captures the setting where a corrupt device 
repeatedly tries to corrupt a randomly chosen neighboring device and in 
each trial it successfully corrupts the device with tiny probability. 
%It has been shown that 
Such a model approximately captures the true
propagation  of a malware~\cite{wang2003epidemic,zou2002code}.

Alternatively, an external adversary corrupts a device by directly
connecting to it, and not through one of its neighbors. Specifically, \adv\ first chooses a random device 
from the network and waits for a time drawn from an exponential
distribution with parameter $\lambda_\extl$ and corrupts the chosen
device. \adv\ repeats this till it is forcefully disconnected from
the network. Such an adversary captures proximity attacks where 
the attacker enters the wireless range of the victim and corrupts
it~\cite{ronen2017iot}. 

\subsection{Hardware Modules}
\label{sub:hardware}
We next describe the memory organization and communication
requirements of
%we consider for 
devices in \dante.
%in our network. 
This memory organization is already considered in 
embedded trust anchors such as SMART~\cite{eldefrawy2012smart}, 
TyTAN~\cite{brasser2015tytan}, and TrustLite~\cite{koeberl2014trustlite}. 
TrustLite and TyTAN have been implemented on Intel's Siskiyou Peak 
research architecture~\cite{trustile2018impl}. 
% Also, implementation of SMART anchors are available~\cite{eldefrawy2012smart}.

\vspace{0.5mm}
\noindent{\bf Memory Organization.}
The memory of each device is divided into four parts where each part
serves a distinct purpose and has different access control. 
The first part is the Read-Only Memory~(\rom) whose contents are fixed 
during manufacturing and are independent of the application running
on the device. \rom\ stores all procedures required for secure 
execution of the \dante\ protocol. \rom\ is executable, and its 
contents are publicly accessible. One crucial thing to note is that procedures present in \rom\ can be only invoked starting at  designated pre-specified entry points and are executed atomically 
without any interrupts.

The remaining memory regions are non-volatile and are divided into
three parts: \code, \data, and \sram\ region. \code\ region is 
executable, stores all application binaries, i.e., $A_i$ of device 
$n_i$. The  \data\ region is non-executable and is used to store data, and 
run-time environments such as stack and heap for both procedures 
in both \code\ and \rom\ regions. Hence, \adv\ can run modified 
binaries only if they are stored in the \code\ region.

Lastly, \sram~(Secure RAM) is non-executable and is inaccessible to 
procedures present in \code. \sram\ is used to store mission-critical
mutable data that needs protection from \adv. 
We achieve this property using {\em Execution Aware Memory Access 
Control}~(EA-MAC) introduced in SMART architecture 
and later improved by TrustLite~\cite{koeberl2014trustlite}. 
EA-MAC enforces read/write controls depending upon the 
address of the instruction that is currently being executed. 
During secure boot of a device, EA-MAC allows a user to specify tuples 
of memory range say $(c,m)$ with the semantics that memory range $m$ 
can only be accessed by instructions present in memory range $c$. For 
example, TrustLite achieves EA-MAC through its {\em Memory Protection 
Unit}~(MPU). Figure~\ref{fig:memory-all} summarizes the memory
organization of \dante\ along with their access permissions.
% Hereon, we assume that each device
% in \dante, independent of LR or HR has minimal hardware support
% to provide EA-MAC. 

\begin{figure}[!h]
    \captionsetup[subfigure]{aboveskip=-4pt,belowskip=-4pt}
    \centering
    \begin{subfigure}{0.14\linewidth}
    \centering
    \includegraphics[width=\textwidth]{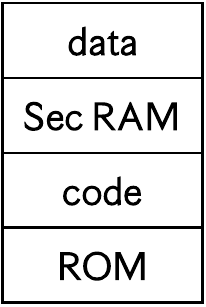} 
    \label{fig:memory-layout}
    \caption{}
    \end{subfigure}    
    \begin{subfigure}{0.70\linewidth}
    \centering
    \includegraphics[width=\textwidth]{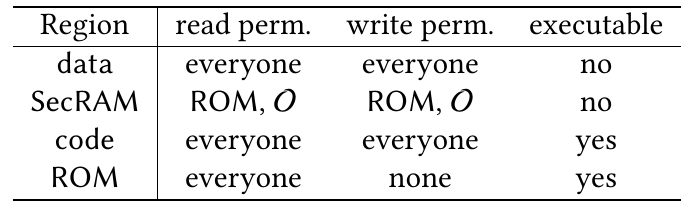} 
    \label{fig:memory-access}
    \caption{}
    \end{subfigure} 
    \caption{Figure~(a) illustrates the memory layout of each
    device. Table~(b) summarizes the read-write permission
    of each memory region. For example, read and write
    access to \sram\ is only given to network operator \opt, and the 
    code present inside \rom.}
    \label{fig:memory-all}
\end{figure}

\noindent{\bf Communication Stack.}
In \dante, we require blank devices to communicate. 
We achieve this using the fact that low-end devices are often equipped 
with separate micro-controller for the networking stack. For example,
SimpleLink Wi-Fi CC3000 connectivity module from Texas Instruments 
is one such commercially available micro-controller with separate 
networking stack. CC3000 supports IEEE 802.11~${\tt b/g}$ 
and has an embedded IPv4 TCP/IP stack~\cite{msp2018dang}. 
As networking stack can be isolated from the underlying application
or operating system that the device is running, we assume that the
networking micro-controller and the software inside it remain 
functional in a \blank\ device and communicates as follows.

\begin{figure}[h!]
    \centering
    \includegraphics[width=0.85\linewidth]{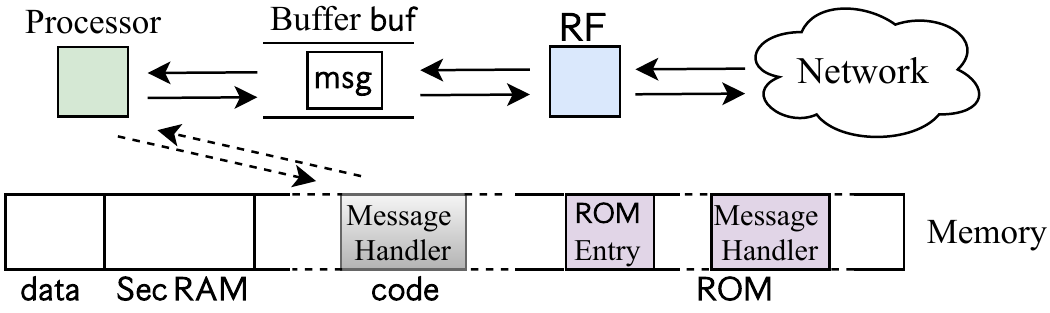}
    \caption{Proposed architecture of a device for enabling a 
    \blank device to communicate with the network. The gray message
    handler belongs to the running application and can be corr. 
    However, the message handler stored inside the
    \rom\ is immutable.}
    \label{fig:arch}
\end{figure}
Let \rf~(Network Interface Card)
denote the micro-controller managing the networking 
stack. When the device is honest, all incoming messages are 
first handled by the Message Handler of the application or the
operating system~(OS). Then it is the responsibility of the 
application's message handler to invoke procedures from \rom\ 
whenever a message is intended for \rom. However, in a blank
device, procedures in \rom\ directly communicate with the 
\rf\ module using the secure message handler stored inside 
\rom. Figure~\ref{fig:arch} illustrates this  
architecture.

\subsection{Connectivity and Network Requirements}
\label{sub:connectivity}
For any given application $b$, the basic requirement of \dante\ is that 
a \blank\ device for $b$ should be able to heal itself as long as there
exist one honest device in the network that has $b$. A necessary condition to achieve
this is that the induced sub-graph formed by devices with code 
of $b$ and devices through which they can communicate is connected. 
To see why, consider the example in Figure~\ref{fig:impossibility}
containing three devices $\{n_1,n_2,n_3\}$.
Devices $n_1$ and $n_3$ are HR devices running application $b_1$ and
device $n_2$ is an LR device running application $b_2$. In this network,
if $n_1$ gets corrupted, it cannot correct itself despite the 
presence of a correct application code at $n_3$. This is because 
device $n_2$ will refuse to forward messages from $n_3$ to $n_1$.
\begin{figure}[ht]
    \centering
    \includegraphics[width=0.50\linewidth]{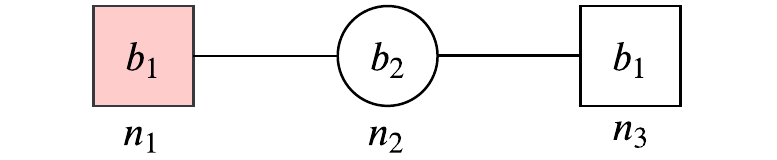}
    \caption{A network of three devices $\{n_1,n_2,n_3\}$ where $n_1$ 
    and $n_3$ are HR devices running application $b_1$ where $n_1$ is 
    corrupted by \adv. $n_2$ is a LR device running application $b_2$.}
    \label{fig:impossibility}
\end{figure}

Generalizing the above, any given network $G=\langle V,E\rangle$ 
must meet the following requirement. Let $V'_{b} \subseteq V$ be 
the set of devices that either runs or stores application $b$. 
Let $H_b$ be the set of HR devices in $G$ that are connected 
to at least one device in $V'_b$ either directly or through 
a sequence of HR devices. Let $V_{b}=V'_{b}\cup H_{b}$ and 
let $G_b\subseteq G$ be the induced subgraph of $G$ with the 
vertex set $V_b$. 
We prove in Theorem~\ref{thm:recover} that \dante\ can correct all 
applications whose $G_b$ forms a connected component and at least 
one honest device that stores program of application $b$ exists in $G_b$ under some specific assumptions. 
In our example in Figure~\ref{fig:impossibility}, $G_{b_1}$ 
consisting of device $n_1$ and $n_3$ is not connected. 
Hence, for \dante\ to be most effective, network
designer must ensure that $G_{b_i}$ for all $i$ are 
connected, which can be achieved by first creating a 
spanning tree among devices running same application
and later add more devices to the spanning tree.
% \vinay{How is a network designer supposed to ensure this? Can we give some guidance on how it can be easily done? Will having a backbone of connected HR devices which acts as a spanning tree for the entire network suffice?}

\subsection{Notations}
Let $|M|$ denote the number of elements in a finite set $M$. 
If $m$ is a integer~(or bit string), then $|m|$ means the 
bit-length of $m$. 
Furthermore, let $\{0,1\}^{\ell}$ denote the set of all bit
strings of length $\ell$. 
% If $E$ denotes some event then $Pr[E]$
% denotes the probability that $E$ occurs. 

% ${\sf f}(x)\rightarrow y$ means that procedure (or algorithm) 
% ${\sf f}$ outputs $y$ given input $x$. We denote an interactive 
% protocol ${\sf prot}$ between devices $n_i$ and $n_j$ by
% $[n_i:x_i; n_j: x_j; *: x_{pub}]\rightarrow [n_i:y_i;n_j:y_j]$ 
% where $n_i$ (resp. $n_j$) is given input $x_i$ (resp. $x_j$) 
% and public input $x_{pub}$. While $n_i$~(resp. $n_j$) is
% operating, it can interact with $n_i$ (resp. $n_j$). 
% As a result of the protocol, $n_i$ (resp. $n_j$) outputs 
% $y_i$ (resp. $y_j$).

\vspace{0.5mm}
\noindent{\bf Attestation.} $v \leftarrow \attest(k, d)$ 
is an algorithm that takes an input $k$, a bit string $d$ 
and computes a deterministic digest $v$ of the $d$. 
Also, \attest\ guarantees \whp~\footnote{For any security parameter
$\ell > 0$, an event happening with high probability \whp\ implies
that the event happens with probability $1-o(1/{\sf poly}(\ell))$. 
Here ${\sf poly}(\ell)$ refers to class of all polynomials with 
parameter $\ell$.} 
that for any pair of keys 
$k, k'$ and data $d, d'$, $\attest(k,d)=\attest(k',d')$
iff $k = k'$ and $d= d'$, where $d$ corresponds to the contents of the \code\ region.

\vspace{0.5mm}
\noindent{\bf Signature.} 
A signature scheme is a tuple of probabilistic polynomial time 
algorithms $({\sf keygen, sign, ver})$. 
$(pk,sk) \gets{\sf keygen(1^{\ell})}$ where $\ell \in \mathbb{N}$. $sk$ and $pk$ are the signing and verification key respectively,
$\sigma \gets {\sf sign}(sk, d)$ is the algorithm to sign string 
$d$ using key $sk$. Lastly, ${\sf ver}(\sigma,d,pk)\in\{0,1\}$ 
is the verification algorithm. 

Unless otherwise stated, throughout the paper we use $\cert(d)$ 
to denote ${\sf sign}(sk_\opti, d)$ and verification of $\cert(d)$
implies verification of the signature using $pk_\opti$, public key 
of the operator \opt.

% !TEX root = ../main.tex
\section{Network Setup}
\label{sec:setup}

\subsection{Device Initialization}
\label{sub:device init}
The \rom\ of each device stores the functions involved in malware
detection and device correction, and is initialized at the time of 
manufacturing. This can be easily extended 
to the setting where contents of \rom\ can be modified using a 
hardware switch present in the device. Hence, \opt\ instantiates 
the remaining regions of memory. 
Executable files of all applications in $A_i$ are stored in the 
\code\ region. \sram\ of each device is initialized with 
$pk_\opti$, a freshly generated asymmetric key 
pair $(pk_i, sk_i)$ unique to $n_i$ along with the certificate
\cert$(pk_i)$. For each application in $A_i$, \opt\ stores 
their version numbers \ver$(A_i)$, \cert$($\ver$(A_i))$, and
\cert$(A_i)$ in \sram. 
The reason behind storing these certificates is to allow the device 
to prove correctness of its code to other devices in the 
network. Lastly, \opt\ also initialize each device with the
self-check rate $\lambda$, the maximum allowable self-check rate
$\lambda_\max$,  and the  minimum allowable self-check rate 
$\lambda_\min$. A detailed description of these self-check 
rates are given in~\textsection\ref{sec:design}.

Once initialized, each device locally generates a cryptographic
symmetric key $k_i$, attest key $ak_i$, and a sequence key $q_i$ 
as uniformly distributed random numbers in $\{0,1\}^\ell$. 
% \vinay{not clear what sequence key is. "Used for authenticating messages" stated below is not clear} $q_i$ as a uniformly distributed random number $q_i\in \{0,1\}^\ell$. 
The attest key is used to compute attestation over the contents 
of \code, and $q_i$ is used to prevent replay attacks. Let $v_i$ 
be the output of the attestation procedure. $n_i$ next generates
a set of keys $L$, of size $\kappa = |L|$, which is used to 
initialize a bloom filter $F$ of size $\mu Z/t$. 
Here $Z$ is the size of the \code\ region whose contents are 
divided into chunks of size $t$ bits each. Refer~\cite{kirsch2006less} 
for more details of bloom filters. 
% Summarizing the above, we 
% formally define device initialization as a protocol \vinay{who runs "init"? Also, in "init" there is a $pk_o$ and a $pk_\opti$. What are both of these?}
% ${\sf init()}$:
% %
% \begin{multline}
% {\sf init}(n_i, pk_o, sk_o) 
%       \rightarrow [pk_\opti;\ (pk_i, sk_i);{\sf cert}(pk_i);\ A_i; 
%       {\sf cert}(A_i);\\
%       {\sf ver}(A_i); {\sf cert}({\sf ver}(A_i));  
%       (ak_i, v_i);\ q_i;\ \{l_j\};\ F;\ (\lambda, \lambda_\min, \lambda_\max)]
% \end{multline}

\subsection{Device Rendezvous} 
\label{sub:rendezvous}
% Let $N'_i$ \vinay{$N'_i$ and $N_i$ seem to be an overkill. Is it that important to separate them? Why not just define $N_i$ as being all neighbors $n_i$ has RDV with?} be the set of devices which are within the transmission 
% range of $n_i$. 
Every device in the network periodically announces itself to 
other devices in its transmission range by broadcasting a 
hello message. On hearing a new device, say $n_j$ with public
key $pk_j$, device $n_i$ rendezvous with it to validate 
each other's certificate $\cert(pk_i)$ and $\cert(pk_j)$. 
On successful validation, they securely exchange their
keys $(q_i,k_i)$ and $(q_j,k_j)$. 
As a device rendezvous with other devices only once, the
key exchange mechanism can be realized using the key-exchange
scheme of TLS 1.3~\cite{bhargavan2017verified}. 
Let $N_i$ be the set of all devices in $n_i$'s transmission range
with whom $n_i$ has rendezvous with, hereon we refer to the devices 
in $N_i$ as the {\em neighbors} of $n_i$. Hence at the end of 
rendezvous, $n_i$ will have a set of $\{k_j,q_j\}_{\forall j\in N_i}$. 
Note that, in our scheme, each device shares the same key with 
all its  neighbors. We do this primarily for efficiency. 
This can be easily extended to establish a unique 
symmetric key between each pair of devices. Table~\ref{tab:content}
summarizes the memory contents of each device at the end of 
initialization and rendezvous.
\begin{table}[h!]
\small
\centering
    \caption{Memory content of a device after initialization and 
    device rendezvous phase.}
    \begin{tabularx}{\linewidth}{c|c|c}
      \hline
      Region & Manufacturing/Initialization & Rendezvous\\ 
      \hline
      \multirow{3}{*}{\sram} & $pk_{\mathcal O}$, $(pk_i,sk_i)$, 
      \ver$(A_i)$, \cert$(A_i)$,
      & \multirow{3}{*}{$\{k_j,q_j\}$} \\
      & {\sf cert}$(pk_i)$, \cert$($\ver$(A_i))$, $\lambda_i$, $\lambda_{\min}$, $\lambda_{\max}$ & \\
      \cline{2}
      & $k_i$, $q_i$, $(ak_i, v_i)$, $F$ & \\
      \hline
      \code & $A_i$ & \\
      \hline
      \multirow{1}{*}{\rom} &\selfcheck(), \attest(), \sign(), \rectify(), \ldots \\
    %   & \bloom, \checksig, \sendcode, \mac, $\ldots$ & \\
      \hline
  \end{tabularx}
  \label{tab:content}
\end{table}

% procedure given in equation~(\ref{eq:rdv}). Note that the output 
% of the \rdv\ procedure is stored in \sram. 
% %
% \begin{multline}
%   {\sf rdv}\left[n_i:\{cert(pk_j),q_i, k_i, sk_i\}; 
%                   n_j:\{cert(pk_i), q_j, k_j, sk_j\};\right.\\
%      \left. *:pk_o\right] \rightarrow   
% \left[n_i:\{q_j, k_{j}\}; n_j:\{q_i, k_{i}\}\right]
% \label{eq:rdv}
% \end{multline}

% Essentially using procedure \rdv, a pair of neighbor devices exchange 
% their symmetric keys and sequence numbers. 

% !TEX root = ../main.tex
\section{Design of \dante}
\label{sec:design}
At a very high-level, correction of a corrupt device in \dante\ 
involves the following steps. 
Each device periodically initiates a self-check procedure to
detect whether it is corrupt or not. In case the device is 
found to be corrupt, its hardware disables execution from its 
\code\ region. Then the device queries its neighbors for a 
correct application code. We next look at each of these procedures 
in detail.
\subsection{Detecting Malware}
\label{sub:detect} 
Every device performs periodic self-check with the time interval
between two consecutive self-checks chosen from an exponential 
distribution with parameter $\lambda$. As expected value
of exponential distribution with parameter $\lambda$ is 
$\frac{1}{\lambda}$~\cite{mitzenmacher2017probability}, the expected time between two consecutive self-check is
$\frac{1}{\lambda}$. We pick time intervals 
from an exponential distribution due to their memoryless 
property~\cite{mitzenmacher2017probability}. As the 
rate of propagation of malware depends crucially on the 
time a device remains infected, memoryless self-checks will
prevent \adv from strategically infecting devices to 
increase the duration for which the device remain corrupt. 
Furthermore, memoryless self-checks prevents a mobile adversary 
from evading detection by uncorrupting a infected nodes
just before the next self-check~\cite{ma2009new}.
Let $\delta\gets\exp\{\lambda\}$ be one such realization of the 
time interval. Starting from last self-check, $\delta$
is decremented by one in every clock cycle. When $\delta$ reaches
zero, the processor generates a hardware interrupt. On this
interrupt, the processor pauses the running application, 
records the run-time state of the application in a non-volatile 
memory and invokes \selfcheck$()$ procedure from \rom. Also, all
interrupts are disabled to allow atomic execution of \selfcheck.

Procedure \selfcheck$()$ first invokes procedure \attest$()$ with its
input as $ak_i$ and entire contents of the \code\ region. 
Let $v'_i$ be the attestation result. If $v'_i$ equals to $v_i$,
i.e., the contents of \code\ are not tampered, \dante\ increments 
the expected wait time between self-checks, that is
$\frac{1}{\lambda}=\frac{1}{\lambda}+1$ as long as it does not 
exceed a pre-defined upper bound. In other words, it sets $\lambda$ 
to $\max\{\lambda_\min, \lambda/(\lambda+1)\}$.
Next, interrupts are enabled and the control is given back to the 
application. 
On the contrary, $v'_i\ne v_i$ implies modification of the 
application code. 
In such a situation, instead of resuming the application, 
\selfcheck\ sets a hardware bit to make the \code\ region 
non-executable, and invokes \rectify$()$, another secure 
procedure from \rom. The pseudocode of \selfcheck$()$ is given 
in  Algorithm~\ref{algo:selfcheck} where we use $[\code]$ 
to refer to the contents of the \code\ region.

\vspace{0.5mm}
\noindent {\bf Malware localization.}
Once the tampering has been detected, the next goal is to identify
the tampered region of the code to avoid downloading the entire 
application program.  
A naive approach of dividing the entire $[\code]$ into chunks of
size $t$ and storing hash of each chunk in \sram\ has high memory
usage. Specifically, if $Z=|[\code]|$, than 
this 
approach would require storing $\ell Z/t$ bits of additional
storage in \sram, where $\ell$ is the size of the output
of the hash function. 

In \dante\ we reduce this storage overhead through novel use 
of bloom-filters. Recall from~(\S\ref{sub:device init}), that a 
$\mu Z/t$ bit long (for small constant $\mu$) bloom filter $F$ 
is initialized with partitions of $[\code]$ using the set
of secret keys $L$. 
Hence, to localize the malware, \rectify$()$ finds all chunks 
that are absent in $F$. The idea is, since \adv\ is unaware of 
keys in $L$, the chunks modified by the adversary will most 
likely be absent in the $F$ and hence will be detected by 
\rectify$()$. For example, in Figure~\ref{fig:bloom-filter}, 
adversary modifies $i^{\rm th}$ chunk to $c'_i$, which is 
absent in the filter $F$. As a result, instead of the entire 
application code  the blank device will query its neighbors 
only for the chunks that are marked as absent in the bloom 
filter.
\begin{figure}[h!]
    \centering
    \includegraphics[width=0.80\linewidth]{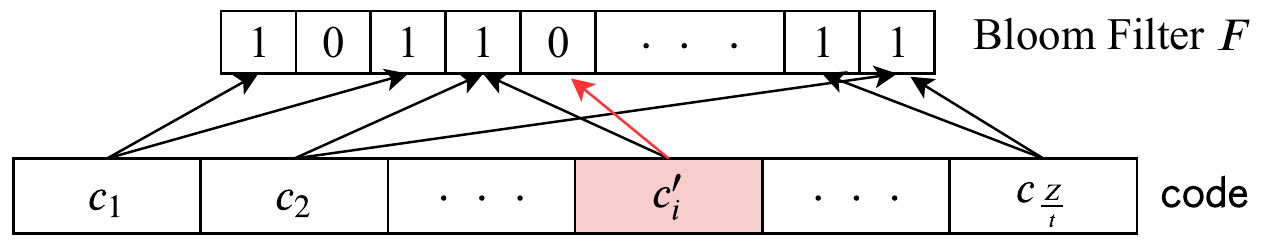}
    \caption{Adversary modifies $i^{\rm th}$ chunk to $c'_i$ which 
    is absent from the filter $F$.}
    \label{fig:bloom-filter}
\end{figure}

% \vinay{We need to be clear what are the advantages of using a Bloom filter. Locally, we are reducing storage by how much? Is it only a constant, or do with have a O() type of gain? Also, when communicating which chunks are missing, is the bloom filter transmitted to save on bandwidth in some way?}
%
% Procedure \rectify$()$ first identifies all possible modified chunks 
% using the bloom filter $F$ of size $\mu Z/t$ and the 
% keys $L=\{l_j\}$. 
% Recall~(ref.~\S\ref{sub:device init}), the 
% bloom filter is initialized on $Z/t$ chunks of size $t$ bits each 
% where $Z$ is the size of \code\ region. Hence, to localize the
% malware, \rectify$()$ finds all chunks that are absent in
% $F$. \vinay{Not clear what you want to say next. Is it that the Adv cannot predict where in the bloom filter the $c_i$'s will point to?} 
% The idea is since \adv\ is unaware of keys in 
% $L$, the chunks modified by the adversary will most 
% likely be absent in the $F$ and will be detected by
% \rectify. 

% As a result, the device will query 
% its neighbors only for these chunks instead of the entire 
% application code. For example, in Figure~\ref{fig:bloom-filter},
% adversary modifies $i^{\rm th}$ chunk to $c'_i$, which
% is absent in the filter $F$.

However, since bloom filters have non-negligible false positive 
rates, it is possible (albeit rarely) that \rectify$()$ fails to 
identify all the modified chunks. In such situation, the blank 
device downloads the entire application program. Also, we we 
keep $t$ and number of keys in $L$ parameterizable that 
can be picked for any desired false positive rate. For example, 
with $Z=16384$, i.e., 16KB of executable memory, which is typically 
the case in MSP430 micro-controllers~\cite{msp2018dang}, 
$\mu=8$, $|L|=4$, and $Z/t=32$, we show in \S\ref{sub:eval performance}
that a blank device will download the entire application program
less than $2\%$ of the time. 
Refer to~\cite{mitzenmacher2017probability} for detailed 
analysis of false-positives in bloom filters.

Once the corrupted chunks has been identified, interrupts 
are re-enabled. We call a device with disabled execution 
as a {\em \blank} device. Recall~(ref.~\S\ref{sub:hardware}), 
in all \blank\ devices all incoming messages are directly 
handled by functions in \rom~(ref.~\textsection\ref{sub:hardware}).

\subsection{Correcting \blank\ Devices}
\label{sub:correction}
The basic idea of correction is that once a device $n_i$ becomes
\blank, it asks one of its neighbors to send the correct
version of the compromised code along with the  
certificate from \opt.  $n_i$ on receiving these chunks validate
their correctness by checking the certificate from \opt. On 
successful validation, it installs them in its \code\ region and 
starts normal execution of the application program. Further, in
scenarios where devices in $N_i$ are running different versions of the
code, it is desirable to download the most recent version of the
application among all available versions. Here we are implicitly 
assuming that recent versions of application programs have a 
higher version number.

A na\"ive approach is to send a message to each neighbor and
request for the necessary chunks of code. On receiving the 
application programs from each neighbors, $n_i$ locally identify 
the highest version, validates it and then installs it. 
This approach is bandwidth inefficient as it 
requires each neighbor to transmit all codes, which might 
be relatively large in a resource constrained setting. 

An alternate approach is to first ask neighbors for the 
version number of application $b$ they are running, 
and then request the neighbor running the 
highest version to send the code. Although this 
approach is bandwidth-efficient, it has several limitations. 
First, this approach does not protect $n_i$ from requesting code from 
a malicious neighbor that might deny or delay the response to the 
code request by merely dropping or delaying the code request message.
Further, in the case of dense network the cost of transmitting 
so many version messages could still be overwhelming.
Also, none of these approaches prevent a corrupt device from 
sending spurious version and code request to honest devices 
and drain their bandwidth and computation resources. 

\vspace{0.5mm}
\noindent{\bf Our Approach.}
Let $\Pi \subseteq [Z]$ denote the set of corrupt chunk indices
at the \blank\ device $n_i$. For simplicity, let us assume that all
of the corrupted chunks belong to a single application $b$ with its
version being $z_i=\ver(b)$. Also, let us assume that $\Pi$ includes 
all modified indices, i.e., there is no false positive due to the 
bloom filter.
% \vinay{I need to understand what you mean be false positive to understand this sentence. Also, can $\Pi$ include indices which are not bad?} 
Let $N^{(b)}_i \subseteq N_i$ denote the set of devices among neighbors
which are in $G_{b}$, i.e., the induced subgraph of $G$ for application 
$b$~(ref.~\textsection\ref{sub:connectivity}). 
% Note that, in addition to all LR device storing $b$, $N^{(b)}_i$ also 
% includes HR devices that do not store or run $b$ themselves but are
% connected to a device that stores/runs $b$ either directly or through 
% a sequence of HR devices. 
We assume that $n_i$ is unaware of the identities of devices in $N^{(b)}_i$.

To request correct code, $n_i$ broadcast to its neighbors a message 
$\msg_{\req}$ with $\langle \req, \ttl, q_i,|N_i|,z_i,b,\Pi \rangle$ as 
its payload. Unless otherwise stated, we assume that all messages are 
tagged with a message {\em Message Authentication Code}, source of 
messages can be established for every message transmitted in the 
wireless range, and sequence number $q_i$ is incremented by $n_i$
after every message. Tag $\req$ in message payload specifies that 
this message is to request for binaries. Sequence number $q_i$ 
assists devices in $N_i$ to establish validity and freshness of
$\msg_{\req}$. 

\vspace{0.5mm}
\noindent{\bf Adaptive self-check rate.}
% \vinay{We need to first state what ttl is. Have we done that? Otherwise the following is very confusing.}
Each honest device $n_j \in N_i$ on receiving $\msg_\req$ first updates 
its self-check rate as:
\begin{equation}
    \lambda \gets \mathbb{1}_{\ttl>0} \min \{ 2\lambda, \lambda_\max \} + \mathbb{1}_{\ttl \le 0} \lambda
    \label{eq:rate update}
\end{equation}
where $\mathbb{1}_x$ is a indicator function which is equal to 
value 1 if $x$ is true and 0 otherwise. $\ttl$ in the message 
payload is the parameter to limit broadcast of device corruption
message. 
Additionally, when $\ttl>0$, device $n_j$ broadcasts a warning 
message to all its neighbors, i.e., devices in $N_j$ informing 
about corruption of $n_i$ with parameter $\ttl-1$. Similar to 
devices in $N_i$, devices in 
$N_j\setminus N_i$ updates their self-check rate according to 
equation~\ref{eq:rate update} and recursively forwards it to their 
neighbors as long as $\ttl$ reaches zero. 
Figure~\ref{fig:adaptive-self-check}~(a) and~(b) illustrates the
self-check rate of neighbors of $n_3$, before and after $n_3$
broadcasts $\msg_\req$ with $\ttl=1$.
\begin{figure}[h!]
    \centering
    \includegraphics[width=0.95\linewidth]{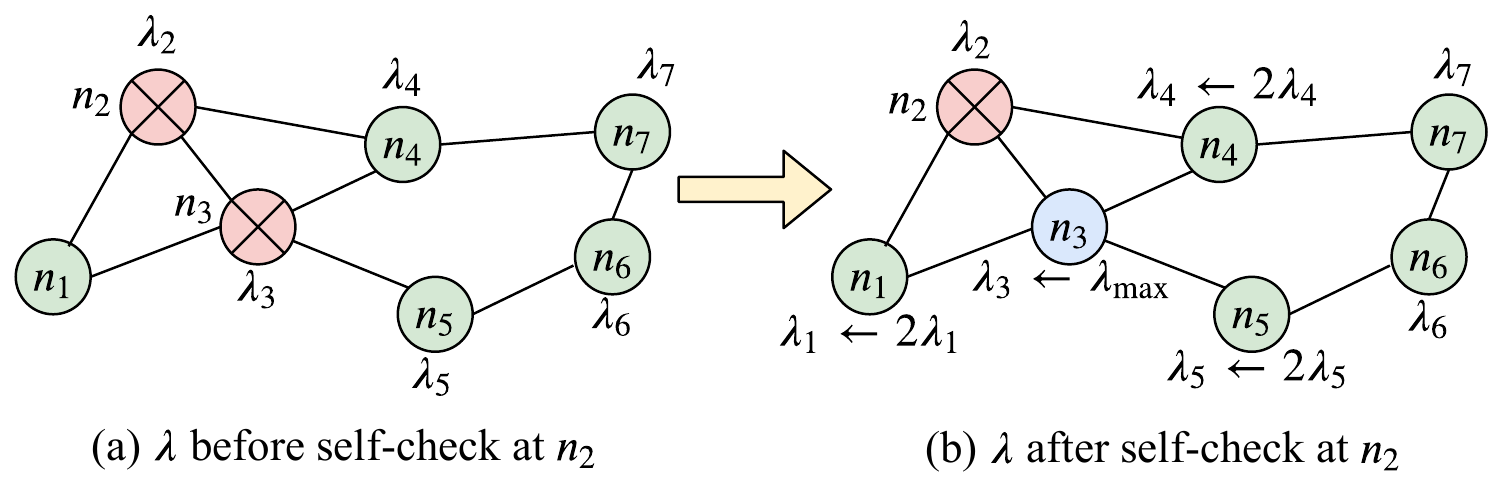}
    \caption{Self-check rate of neighbors of a device $n_3$
    before~(a) and after~(b) $n_3$ broadcasts $\msg_\req$ with
    $\ttl=1$. As a result, honest neighbors $n_1,n_2$ and $n_5$ 
    updates their $\lambda$ using equation~\ref{eq:rate update}.
    Also, $n_3$ sets its own self-check rate to $\lambda_\max$.}
    \label{fig:adaptive-self-check}
\end{figure}

\vspace{0.5mm}
\noindent{\bf Code transmission with random-backoff}.
To address the issue of redundant code transmission, each 
neighbor $n_j$ of a blank node $n_i$, performs a uniform 
random-backoff with backoff delay $\tau_j$ as:
\begin{equation}
    \tau_j =  \max\{\Delta - (z_j-z_i), 0\}|N_i|\theta + 
    \left\lfloor{\mathcal U}(0,1)|N_i|\right\rfloor\theta 
    \label{eq:back-off}
\end{equation}
%
% \vinay{next sentence not clear.}
% \vinay{what if a malicious neighbor chooses its $\tau$ to be zero and immediately starts transmitting. Will that cause a DoS attack?}
where $\Delta$ estimate of maximum difference in version numbers 
among devices running a particular application. Similarly, 
$\theta$ is a protocol parameter denoting the approximate 
upper bound on time required to transmit the requested chunks, 
$z_j$ is the version number of $b$ at $n_j$, and 
$\mathcal{U}(0,1)$ is a value chosen uniformly randomly between
$(0,1)$. The intuition behind
this approach is two fold: {\em first}, we prioritize responses 
from devices running a higher version of the same application; 
{\em second} among devices running the same version of the
application, we aim to spread the time when these device 
transmits the requested chunks. Device $n_j$ only starts the 
timer if $z_j \ge z_i$, otherwise $n_j$ simply discards the 
message. Figure~\ref{fig:random-backoff} illustrates the 
the distribution of transmission time at neighbors of $n_1$.
Pseudocode in~\ref{algo:handle-req} describes the steps taken
by each device in $N_i$.
\begin{figure}[h!]
    \centering
    \includegraphics[width=0.80\linewidth]{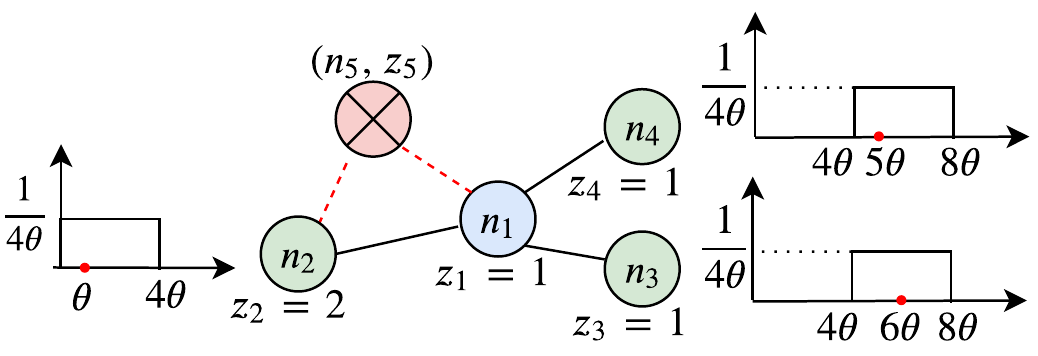}
    \caption{Distribution of time at which neighbors of blank
    node $n_1$, i.e., $N_1=\{n_2,n_3,n_4,n_5\}$ transmits the
    requested chunks of code. Here, $\Delta=1, |N_1|=4$, node $n_5$ 
    is corrupt, version of $n_2$, i.e., $z_2=2$ and all remaining
    node has version $1$. Red dot on the time axis in each graph,
    is one realization of the transmission time.}
    \label{fig:random-backoff}
\end{figure}

Without loss of generality, let $n_j \in N^{(b)}_i$ be the 
honest device with smallest back-off interval $\tau_j$ among 
all honest devices in $N^{(b)}_i$. Once $\tau_j$ expires, 
$n_j$ sends a single chunk to $n_i$ and waits for an 
acknowledgement from $n_i$. On receiving the acknowledgement
message from $n_i$, $n_j$ sends the remaining chunks. If more 
than one honest device simultaneously sends the first chunk, 
$n_i$ sends acknowledgement to only one of them. 
We present a detailed analysis of such scenarios
in~\S\ref{sub:communication cost}.
% \vinay{If more than one starts sending, then why does $n_i$ not tell all except one to stop sending? } sends the requested chunks at the same
% time, $n_i$ installs them only once. 

\noindent{\bf Stream Signatures.}
If we use a signature scheme in which a blank device $n_i$ must
receive all chunks before verifying their signatures, it will
allow an adversary to waste a lot of bandwidth by sending 
invalid chunks and $n_i$ will not know they are invalid till 
the very end. 
We mitigate this attack using on-line variant of stream 
signature introduced in~\cite{gennaro1997sign}. In stream 
signature, the signer \opt, signs first chunk and embeds 
in each chunk $c_i$ the hash of the next chunk $c_{i+1}$. 
Figure~\ref{fig:stream sig} illustrates this. As a
result, a bogus chunk can be detected immediately.
\begin{figure}[ht]
    \centering
    \includegraphics[width=0.80\linewidth]{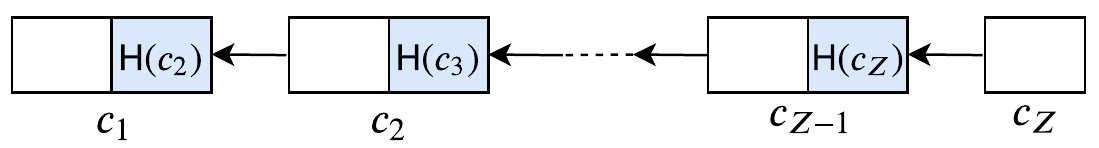}
    \caption{In stream signature messages are divided 
    into chunks and each chunk (except the last chunk) contains
    the hash of the next chunk, i.e., $c_i$ contains ${\sf H}(c_{i+1})$. The signer only signs the $c_1$.}
    \label{fig:stream sig}
\end{figure}

Lastly, once $n_i$ receives all chunks in $\Pi$, it broadcasts a 
\msg$_{\sf done}$ to its neighbors indicating that it has successfully 
corrected itself. Honest neighbors on hearing \msg$_{\sf done}$ cancel
their back-off timers (if any) corresponding to $n_i$'s code request.
Alternatively, if $n_i$ do not receive all the correct chunks 
within time $\Delta|N_i|\theta + |N_i|\theta$, $n_i$ rebroadcasts
\msg$_\req$ with the updated $\Pi$ after a time delay of 
$\delta$ drawn from $\exp\{\lambda\}$. Such as situation 
could possibly arise if either all devices in $N^{(b)}_i$ are
running a lower version of $b$, or they are in \blank\ or 
corrupt state. Algorithm~\ref{algo:handle-resp} presents the
pseudocode for handling a response to $\msg_\req$ message.

\vspace{0.5mm}
\noindent{\bf Fast Correction.}
\dante\ also enables fast correction of a cluster of \blank\ devices.
With solely the method just described above, if the nearest honest device is $r$ hops away, correction of $d_i$ takes
in the best case an expected time of $r/\lambda_\min$. To enable
faster correction, whenever a device $n_j$ is corrected, it 
immediately broadcasts a message containing information about
the corrected code, its version number, and the corresponding 
certificates. On hearing this message, \blank\ device seeking the 
appropriate binaries can actively request it from device $n_j$.
As a result, the corrected binaries spreads through the network much 
faster without waiting for the timers of \blank\ devices to expire.

% !TEX root = ../main.tex
\section{Update of Application Binaries}
\label{sec:code update}
\begin{figure*}[ht]
	\centering
	\includegraphics[width=0.95\linewidth]{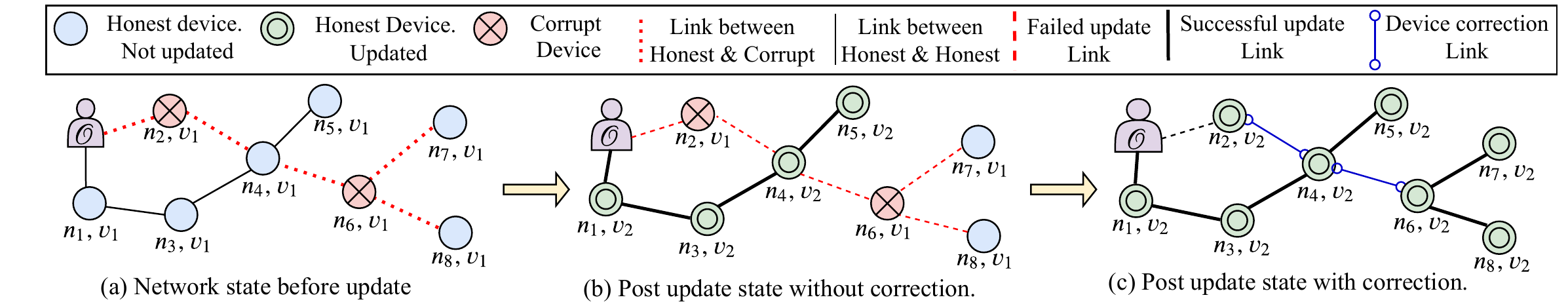}
	\caption{code update in a network of 8 devices $\{n_1,\ldots,n_8\}$
	with device $n_2$ and $n_6$ corrupt prior to update~(a). 
	Initially all device run the same version $v_1$ of the application.
	Let $v_2>v_1$ be the updated version, then without a correction 
	mechanism only $\{n_1,n_3,n_4,n_5\}$ will be updated~(b). However, when
	deployed along with \dante\, the entire network will get updated
	once $n_2$ and $n_6$ perform self-check~(c).}
	\label{fig:update}
\end{figure*}
So far we have only looked at how a compromised device self-corrects
itself with the help of its neighbors. We now consider the behavior of
the whole network that is running the \dante protocol specifically in
situations where \opt updates the application program executed with
newer versions. To expound \dante's applicability for updating 
binaries in a network of low-end devices we consider a prevalent 
update technique motivated from~\cite{seshadri2006scuba,
kohnhauser2016secure,asokan2018assured}, study it in 
our threat model and show its limitations. We then make
minor modifications
% \vinay{saying "minor" modification may give the impression that the contribution of update is also minor. Throughout the paper we have sidelined the "updating".} 
to the \dante protocols described so far, and show 
that \dante when combined with this network update technique
overcomes these limitations. For brevity we will only focus on 
the network $G_b$ for a specific application $b$. This can be easily 
extended to the entire network. Also, each newer version comes with
a monotonically increasing version number.

Consider the following recursive swarm update mechanism 
of~\cite{asokan2018assured}. Here \opt\ first must find
one device which is honest and then update it with a newer version of
the application. Finding an honest device is important because a corrupt
device can simply drop the update messages. This originator device 
then updates its neighbors and so on to form a virtual update tree. 

This scheme has several shortcomings. First, in case a large fraction 
of devices are corrupted, \opt may have to contact many devices to 
find one honest one. Hence it allows malware to spread for longer
duration. Second, the above approach can only update devices connected
to the \opt\ through a sequence of honest device and all other 
devices may still remain corrupted. 

Suppose we use \dante along with this recursive update
procedure. Even if a large number of devices are corrupt, they will
become blank and then get corrected over time. Thus \opt has a higher
chance of encountering a device which is either blank or running a
correct application. In fact, if \opt is in contact with $\kappa$ 
% \vinay{are these devices all corrupt?} 
devices (say in wireless communication range of them), then from 
elementary probability theory, the expected time for at least one 
of them to become blank is at least $1/((\kappa+1)\lambda_\min))$. 

We now show how the second problem of the update scheme 
in~\cite{asokan2018assured}, namely the 
% \vinay{``inability'' of which scheme?} 
inability of the update
to reach any corrupted device, is solved. The update propagates on an
honest virtual tree as before. Consider a corrupt device which has a
neighbor in this tree. After it performs a self-check it becomes
blank and then obtains the latest version from its neighbors.  We now
make a minor modification to \dante. This device then acts like a new
root and propagates the latest version to its neighbors who are
honest but do not have the latest version. This increases the size of
the virtual tree running the latest version, until the virtual tree
encompasses the entire $G_b$. This is illustrated in the transitions
of Figure~\ref{fig:update}.

Let the network shown in Figure~\ref{fig:update}(a) be the $G_b$
consisting of devices $\{n_1,\ldots,n_8\}$ for application $b$ with $v_1$
as the current version. Let $n_2$ and $n_4$ be the corrupted devices.
With this initial state of the network, \opt\ will successfully
initiate the update procedure with at-most two trials. 
Let $v_2>v_1$ be the newer version of $b$. Without any correction
mechanism the update will fail to reach honest device $\{n_7,n_8\}$. 
Also the corrupt devices will not be updated as well. 
Figure~\ref{fig:update}(b) illustrates this.

However, in \dante, as soon a corrupt device performs a self-check and 
detects that it has been compromised, it will download the updated 
code from one of its neighbors. Further, it will forward the 
information about the newer update to all the devices in its 
neighborhood that are running an obsolete version of the application.
Stated differently, when $n_6$ corrects itself it then behaves
as a new originator and updates $n_7$ and $n_8$ as shown in  
Figure~\ref{fig:update}(c). 
This is analogous to a temporary pause of the original update
procedure due to adversarial devices in the path and its resumption
later as the devices enter the blank state as a part of the
protocol.

% !TEX root=../main.tex
\section{Analysis}
\label{sec:analysis}
%
% In this section we 
% rigorously analyze these improvements over naive solutions.
% Next, we theoretically prove that \dante\ recovers entire 
% network from adversarial corruption under specific assumptions.

\subsection{Secure Memory Cost}
\label{sub:memory cost}
Recall~(ref.~\S\ref{sub:device init}), for malware localization,
\dante\ uses a bloom filter of size $\mu Z/t$ and $|L|$ keys for
input to the hash function of the bloom filter. Hence, \dante\ 
stores $\ell|L|$+$\mu Z/t$ bits of information in secure memory.
Where as, naive approach of storing hash of each chunk would 
have required $\ell Z/t$ bits of memory. 
Next, with the help of  Table~\ref{tab:content}, we evaluate the 
size of \sram\ required to store the remaining information for a 
single application. This can be easily extended to multiple
applications. Each device $n_i$ stores two asymmetric public 
key, $pk_{\mathcal O}$ and $pk_i$, one asymmetric private key $sk_i$. 
Let $|pk_{\mathcal O}|=|pk_i|=|sk_i|=1024$. If each certificate is of size 
256 bits, each device stores one certificate for its public key 
and two certificates for each application. 
Also we use same number of bits for all three self-check parameters, 
i.e., $|\lambda|=|\lambda_\min|=|\lambda_\max|=32$. Similarly, let
$\ell=|k_i|=|q_i|=|ak_i|=|v_i|=128$. 
Lastly, for each of its neighbor in $N_i$, a device needs to store
$256 = 2\times 128$ for the shared symmetric key and the sequence
number. Summarizing the above,
\begin{equation}
|{\sf SecRam}| = 3|pk|+(4+|L|)\ell + 3|{\sf cert}|+ \frac{\mu Z}{t} + 3|\lambda| + 2|k||N_i|
\end{equation}

% we divide the \code\ region
% of size $Z$ into chunks of size $t$. Next, we initialize the bloom
% filter of size $\mu Z/t$ bits with each chunk using $|L|$ keyed 
% hash functions. For small $|L|$ this gives us significant saving 
% over a naive scheme of storing hash of all $Z/t$ chunks in 
% \sram. If, $\ell$ denote the size of each hash output, note that 
% bloom filter reduces this memory usage by a factor of $\ell$. Summarizing the above, we have:
% %
% \begin{table}[h!]
%   \begin{tabular}{cccc} 
%       \hline
%       Approach & |code| & chunk size & \sram\ usage \\ 
%       \hline
%       Naive & Z & t & $\ell Z/t$ \\ 
%       Bloom Filter & Z & t & $\ell|L|$+$\mu Z/t$ \\ 
%       \hline
%   \end{tabular}
%   \label{tab:permmission}
% \end{table}

\subsection{Communication Cost.}
\label{sub:communication cost}
The first major source of communication is due to the fact that
a blank device in \dante\ only requests for the modified chunks. 
%
% Although procedure~\attest$()$\ only assist us in detecting the 
% presence of malware in the device \whp, it is the usage of bloom 
% filter which assist us to localize the malware from the entire \code\ 
% and hence it obviates the need for downloading the entire application.
%
However, as bloom filter has non-negligible false-positive rates
and it is possible (albeit rarely) that the bloom filter fail to 
localize the malware. Thus for any given $Z$, $t$, $\mu$ and $L$, 
we compute the expected number of chunks a blank device needs to 
download to correct itself. 

Let $\kappa$ be the number of chunks modified by the adversary.
Since, the bloom filter keys are inaccessible to the attacker, 
from elementary cryptography, the attacker can not make strategic
modifications to evade the bloom filter check~\cite{kirsch2006less}. 
Hence, we assume modifications of these chunks to be arbitrary. 
Let $p$ be the false positive rate for a single chunk; then with 
the above assumptions, $p=(1-e^{-|L|/\mu})^{|L|}$. Refer to~\cite{kirsch2006less} for more details.
% \begin{equation}
% p=(1-e^{-|L|/\mu})^{|L|}
% \label{eq:false positive}
% \end{equation}

\begin{theorem}
Assuming hash functions are ideal, if an adversary corrupts $\kappa$ 
chunks from a total of $Z/t$ chunks in a device which uses bloom 
filter scheme of~\cite{kirsch2006less} with $\mu Z/t$ bit 
filter and $|L|$ hash functions, then the probability that the blank
device download the entire code is:
\begin{equation}
\Pr[\text{download entire code}] = 1-\left(1-p\right)^\kappa 
\label{eq:all code}
\end{equation}
Also, expected number of chunks the blank device will download is:
\begin{equation}
  \frac{Z}{t}\left(1-\left(1-p\right)^\kappa\right) + \kappa\left(1-p\right)^\kappa
  \label{eq:expected chunk} 
\end{equation}
\end{theorem}
\begin{proof}
% \vinay{I could not understand the proof}
Whenever, the device all $\kappa$ modified chunks, it only
downloads $\kappa$ chunks. This gives us the second term of
equation~\ref{eq:expected number}. Alternatively, even with
a single false positive among $\kappa$ chunks, the device 
downloads all $Z/t$ chunks. Combining this with 
equation~(\ref{eq:all code}), we get the first term of 
our result.
% The proof follows from the fact that the device needs to download 
% the entire code, even if one of the adversarial corrupt chunk 
% results in a false positive in the bloom filter. Combining this 
% with equation~\ref{eq:false positive} results in the probability
% that a device needs to download the entire code. Similarly, 
% using equation~\ref{eq:all code}, it is easy to get the expected
% number of chunks given in equation~\ref{eq:expected chunk}.
\end{proof}

The next source of communication improvement is due to the 
random back-off procedure used for reducing the number of 
neighbors that transmit the requested chunks. The following
theorem~(proof in Appendix) illustrates that the expected 
number of neighbors that will transmit the requested chunks. 
\begin{theorem}
If a device has $m$ neighbors, then the expected number
of neighbors that transmits the requested chunks are
\begin{equation}
\sum_{j=1}^{m-1}\sum_{k=1}^m \left(\frac{k\binom{m}{k} (m-j)^{(m-k)}}{m^m}\right)+\frac{m}{m^m} 
\label{eq:expected number}
\end{equation}
\label{thm:expected number}
\end{theorem}

\subsection{Recoverability}
\label{sub:recoverablity}
Next, we theoretically argue that \dante\ recovers and guarantees update of the entire 
network in the presence of both internal and external adversary under specific
assumptions~(Proofs in Appendix~\ref{sec:proofs}).
For an heterogeneous network $G=\langle V,E\rangle$ of devices, 
we define the graph $G_b \subseteq G$ for application $b$ as:
\begin{definition}
For any given application $b$, let $V'_{b} \subseteq V$ be the 
subset of devices that either runs or stores the application $b$. 
Let $H_b$ be the set of HR devices in $G$ that are connected 
to at least one device in $V'_b$ either directly or through 
a sequence of HR devices. Let $V_{b}=V'_{b}\cup H_{b}$. 
Then $G_b$ is the induced subgraph of $G$ due the vertex set $V_b$. 
\label{def:app graph}
\end{definition}

\begin{theorem}
If $G_b$ is connected and no additional device gets corrupted
after a given time $t_0$ and there exits at least one honest 
device running or storing application $b$ at time $t_0$, 
then \dante\ corrects all devices in $G_b$.
\label{thm:recover}
\end{theorem}

\begin{theorem}
If $G_b$ is connected and if the update patching the vulnerability
is successfully initiated by \opt\ in at least one device in $G_b$,
then \dante\ guarantees update of the entire network in the presence of
both internal and external adversary. 
\label{thm:update}
\end{theorem}
% !TEX root = ../main.tex
\section{Simulation}
\label{sec:simulation}
Since the cost of evaluating \dante\ on a large scale network 
consisting of thousands of device would be high, we test \dante\
by simulating it in \omnet\ version ${\sf 5.5.1}$~\cite{OpenSim}. 
We simulate both internal and external malware propagation with
update scheme of~\cite{asokan2018assured}.

% \subsection{Network Topology}
% \label{sub:network}

%
\vspace{0.5mm}
\noindent{\bf Network Topology.}
We test \dante\ on three different topology with approximately
1024 LR devices each, with all devices running the same application. 
Our first topology is a connected {\em \adhoc} wireless 
network of 1024 devices spread uniformly across an area of 4~km$\times$4~km. Each device has a wireless transmission range 
of 200 meters around it. The intent behind this topology was to 
capture scenarios such as the ad-hoc deployment of sensor network 
that are ubiquitous in Military application, agriculture, forest 
fire monitoring system, etc.~\cite{vasisht2017farmbeats,johnsen2018application,
jalaian2018evaluating}. The remaining two topologies we simulate 
are Binary and Ternary tree. We pick them to capture Industrial IoT, Building management etc~\cite{dhondge2016hola}. 
In all the above topologies, we use the same 20~ms average 
transmission delay between each pair of connected devices, as it is 
the average value in ZigBee sensor networks~\cite{spanogiannopoulos2009simulation}. 
Lastly, during an update, \opt\ connects to a randomly chosen 
device and update it with a newer version.

% \subsection{Internal Adversary}
% \label{sub:sim-internal}
%
\vspace{0.5mm}
\noindent
{\bf Internal Adversary.}
To evaluate the effect of internal adversary \adv$_\intl$, for each 
the topology we corrupt $f=30\%$ of the randomly chosen devices to begin with. We also vary the malware propagation rate, $\lambda_\intl$ and 
the number of hops in the limited broadcast to inform neighboring 
device about the presence of an adversary in the network. 
For each topology, we consider two different initial configuration 
depending upon the positioning of the corrupt devices. Namely, we
consider configuration $C_0$ and $C_1$. In $C_0$, the initial 
fraction of corrupt devices are distributed uniformly randomly across 
the entire network. In $C_1$, the corrupt devices form a single 
island, i.e., corrupt device form a single connected network. To 
create these initial configurations, we first enumerate all the 
device. For configuration $C_0$, we then pick $f|\nodes|$ unique device 
uniformly randomly. For $C_1$, we first select a device uniformly 
at random and starting at this chosen device; and then we pick up to 
$|\nodes|f$ device by performing breadth-first-search. 
% We present evaluation results for $\lambda_\intl=\lambda_\max, 2\lambda_\max$.
%
% Recall from~\textsection\ref{sub:adversary}, in internal adversary,
% each corrupt device repeatedly spreads the malware in the following 
% manner. The corrupt device first chooses one of its neighbors at random. 
% Then it corrupts the chosen device after a delay chosen from an exponential distribution with parameter $\lambda_\intl$. 

% \subsection{External Adversary}
% \label{sub:sim-external}
%
\vspace{0.5mm}
\noindent
{\bf External Adversary.}
External adversary \adv$_\extl$ corrupts uniformly randomly independent 
of the devices corrupted in the past. Unlike \adv$_\intl$, we 
evaluate the effect of \adv$_\extl$ starting from network with all
honest device. Also, we disconnect, i.e., disallow \adv$_\extl$ from 
corrupting more devices after a specified period. In practice,
one can disconnect \adv$_\extl$ from further corruption by isolating 
it from the internet.
Let $\lambda_\extl$ be the corruption rate of \adv$_\extl$. \adv$_\extl$ corrupts a randomly chosen device 
after intervals drawn from a exponential distribution with parameter 
$\lambda_\extl$. Also, once \adv$_\extl$ is disconnected from the network, no additional device gets corrupt. 
% We evaluate \dante\ with $\lambda_\extl=\lambda_\max, 2\lambda_\max$. 

% \subsection{Correction and Update}
% \label{sub:sim-correction}
%

\vspace{0.5mm}
\noindent
{\bf Correction and Update}
For all our simulations, we use initial $\lambda=1/100$, i.e., the average 
inter-arrival time between two consecutive self-check is 100 seconds. 
To evaluate the network behavior with an adaptive self-check rate, we run
all our experiments with $\ttl=0,1,4$. Note that, $\ttl=0$ is the
baseline situation where neighboring devices do not increase their 
self-check rate on hearing warning messages from their neighbor. In all
these experiments we keep $\lambda_\max$ and $\lambda_\min$ to be 
$1/100$ and $1/400$, respectively.

% !TEX root = ../main.tex
\section{Evaluation}
\label{sec:evaluation}
All the results presented in this section corresponds to simulation 
of \dante\ for 1000 seconds. These results are averaged after
10 simulations with distinct randomness seed. 
Unless otherwise stated, updates in the presence of 
\adv$_\intl$ and \adv$_\extl$ are scheduled at 500 and 700 
seconds respectively, from the start of the simulation. 
\adv$_\extl$ is disconnected at time 300 seconds from 
the start of the experiment.
% We only simulate updates that patches the vulnerability 
% because without the patch an update is equivalent to a scenario 
% with no updates. 

\pgfplotsset{small,label style={font=\fontsize{8}{9}\selectfont},legend style={font=\fontsize{7}{8}\selectfont},height=3.8cm,width=1.2\textwidth}
% !TEX root = ../main.tex
\subsection{Internal Adversary}
\label{sub:eval-internal}
\begin{figure}[h!]
\begin{centering}
\ \ \ \ \ \ \ \ \ref{internal}
\begin{subfigure}{0.47\linewidth}
    \begin{centering}
    \begin{tikzpicture}
    \begin{axis}[
        % legend pos=north west,
        % ymin=0.0,
        ymax=0.33,
        xmax=600,
        ylabel={Fraction of devices},
        xlabel={time (in seconds)},
        mark repeat=10,
        mark phase=10,
        mark size=1.5pt,
        line width=0.5,
        legend columns=6,
        legend entries={U0;,U1;,B0;,B1;,T0;,T1},
        legend to name=internal,
        mark options=solid,
        ]
        \addplot+[black, mark=x]  table [x=time, y=UnDetectMean, col sep=comma] {data/internal-formatted/init30/int-adv-1.0-c-0-30/int-1.0UniRandom0.csv};
        \addplot+[dashed, black, mark=x]  table [x=time, y=UnDetectMean, col sep=comma] {data/internal-formatted/init30/int-adv-1.0-c-0-30/int-1.0UniRandom1.csv};
        \addplot+[red, mark=+] table [x=time, y=UnDetectMean, col sep=comma] {data/internal-formatted/init30/int-adv-1.0-c-0-30/int-1.0BinaryTree0.csv};
        \addplot+[dashed,  red, mark=+] table [x=time, y=UnDetectMean, col sep=comma] {data/internal-formatted/init30/int-adv-1.0-c-0-30/int-1.0BinaryTree1.csv};
        \addplot+[blue, mark=square] table [x=time, y=UnDetectMean, col sep=comma] {data/internal-formatted/init30/int-adv-1.0-c-0-30/int-1.0TernaryTree0.csv};
        \addplot+[dashed, blue, mark=square] table [x=time, y=UnDetectMean, col sep=comma] {data/internal-formatted/init30/int-adv-1.0-c-0-30/int-1.0TernaryTree1.csv};
    \end{axis}
    \end{tikzpicture}
    \caption{Corrupt}
    \label{fig:int-c0-1-corrupt}
    \end{centering}
\end{subfigure}
\hfill
\begin{subfigure}{0.47\linewidth}
    \begin{centering}
    \begin{tikzpicture}
    \begin{axis}[
        % legend pos=north west,
        % ymin=0.0,
        ymax=0.33,
        xmax=600,
        xlabel={time (in seconds)},
        % ylabel= {Fraction of devices},
        % ylabel= {Time (in milliseconds)},
        % ytick = {\empty},
        % extra y ticks={0,0.05,0.10,0.15},
        % extra y tick labels={0,0.05,0.10,0.15},
        % extra y ticks={0,0.05,0.10,0.15},
        % extra y tick labels={0,0.05,0.10,0.15},
        grid=minor,
        mark repeat=10,
        mark phase=10,
        mark size=1.5pt,
        line width=0.5,
        mark options=solid,
        ]
        \addplot+[black, mark=x] table [x=time, y=UnCorrectMean, col sep=comma] {data/internal-formatted/init30/int-adv-1.0-c-0-30/int-1.0UniRandom0.csv};
        \addplot+[dashed, black, mark=x] table [x=time, y=UnCorrectMean, col sep=comma] {data/internal-formatted/init30/int-adv-1.0-c-0-30/int-1.0UniRandom1.csv};
        \addplot+[red, mark=+] table [x=time, y=UnCorrectMean, col sep=comma] {data/internal-formatted/init30/int-adv-1.0-c-0-30/int-1.0BinaryTree0.csv};
        \addplot+[dashed, red, mark=+] table [x=time, y=UnCorrectMean, col sep=comma] {data/internal-formatted/init30/int-adv-1.0-c-0-30/int-1.0BinaryTree1.csv};
        \addplot+[blue, mark=square]  table [x=time, y=UnCorrectMean, col sep=comma] {data/internal-formatted/init30/int-adv-1.0-c-0-30/int-1.0TernaryTree0.csv};

        \addplot+[dashed, blue, mark=square]  table [x=time, y=UnCorrectMean, col sep=comma] {data/internal-formatted/init30/int-adv-1.0-c-0-30/int-1.0TernaryTree1.csv};
    \end{axis}
    \end{tikzpicture}
    \caption{Blank}
    \label{fig:int-c0-1-blank}
    \end{centering}
\end{subfigure}
\caption{Fraction of (a)~corrupt and (b)~blank devices in the presence of 
\adv$_\intl$ with $f=0.30$ and $\lambda_\intl=\lambda_\max$ for configuration
$C_0$. Here B0 and B1 refer to Binary Tree topology with $\ttl=0$
and $\ttl=1$ respectively. Similarly, we use U0,U1 and T0,T1 for 
\adhoc\ and Ternary tree topology respectively.}
\end{centering}
\label{fig:internal-30}
\end{figure}
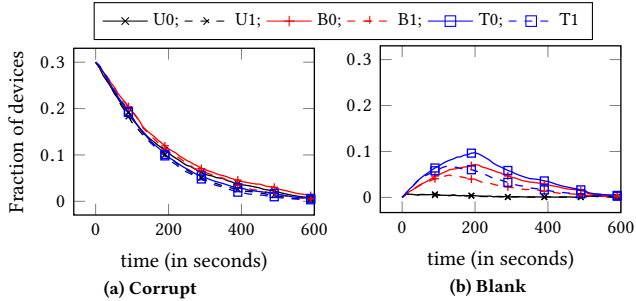
\noindent{\bf Varying Network Topology.}
Figure~\ref{fig:int-c0-1-corrupt} and~\ref{fig:int-c0-1-blank} 
illustrates the fraction of corrupt and blank devices respectively at any 
given time for all three network topologies with configuration $C_0$, 
initial corrupt fraction $f=0.30$, and $\lambda_\intl=\lambda_\max$. 
% Solid and dashed lines corresponds to $\ttl=0$ and $\ttl=1$ 
% respectively. 
Notice that, for all three topologies, the 
fraction of corrupt devices starts decreasing almost from the 
start of the simulation. 
This is because the effective malware spread rate in these 
topologies are lower than $\lambda$. In Binary tree topology, 
approximately half of the initial corrupt devices will be leaf
devices. All these devices only have one neighbor, and also often
this neighbor is shared between multiple corrupt devices.
Hence, these leaf devices will repeatedly try to corrupt an already 
corrupt device. A similar situation arises in Ternary tree
topology as well. Further, the rate of reduction of corrupt devices 
closely follows the tail of an exponential distribution. This is 
due to exponential distribution of interarrival between consecutive
self-checks. 
Notice that the fraction of blank node in \adhoc\ topology remains
almost zero for the entire duration, whereas it first increases in
the tree topologies and then decreases. Again, this is because 
approximately half the devices in tree topologies have only one 
neighbor. Thus these nodes cannot self-correct themselves unless 
their neighbor corrects itself. Also, as expected, we observe a 
lower fraction of corrupt and blank devices for $\ttl=1$ as devices
will perform more frequent self-checks and will recover sooner.

\vspace{0.5mm}
\noindent{\bf Adaptive self-check with varying Configuration.}
Figure~\ref{fig:95-correct} represents the time required for 95\% of
the network to become correct, starting with 30\% of devices being
corrupt for varying $\ttl=1,2,4$. 
Devices in $C_0$~(solid lines) correct themselves quickly than devices 
in $C_1$~(dashed lines). This is because in $C_1$, at any given time,
only devices positioned at the edge of the corrupted island can 
correct themselves whereas devices positioned inside the 
corrupt island need to wait for their neighbor devices get 
corrected. Interestingly, non-zero $\ttl$ introduces a larger
drop in correction time in $C_1$. 
This is because corrupt devices at the boundary of the corrupt
island share a considerable fraction of honest neighbors. Hence, these
honest devices perform faster self-checks as they will update 
their self-check period more frequently. We do not see major 
improvements from $\ttl=1$ to $\ttl=4$ due to the local nature 
of malware propagation.  One exception here is the \adhoc\ 
topology in $C_0$. This was expected as each device in \adhoc\
topology has a higher number of neighbor and hence higher $\ttl$
cautions nodes farther apart to update their self-check rate.
\begin{figure}[h!]
\centering
\hfill\ref{95-correct}
\begin{subfigure}{0.47\linewidth}
    \centering
    \begin{tikzpicture}
    \begin{axis}[
        % legend pos=north west,
        % ymin=200,
        % ymax=300,
        % xmin=480,
        % xmax=600,
        % legend columns=2,
        ylabel= {Time (in seconds)},
        xlabel={$\ttl$},
        grid=minor,
        % mark repeat=10,
        % mark phase=10,
        % mark size=1.5pt,
        line width=0.5,
        % legend columns=3,
        % legend entries={\adhoc, Binary, Ternary},
        % legend to name=internal,
        mark options=solid,
        ] 
        \addplot+[black, mark=x] table [x=ttl, y=c95Mean, col sep=comma] {data/internal-time/init30/int-adv-2.0-c-0-30-no-update/int-2.0UniRandom.csv};

        \addplot+[dashed, black, mark=x] table [x=ttl, y=c95Mean, col sep=comma] {data/internal-time/init30/int-adv-2.0-c-1-30-no-update/int-2.0UniRandom.csv};

        \addplot+[red, mark=+] table [x=ttl, y=c95Mean, col sep=comma] {data/internal-time/init30/int-adv-2.0-c-0-30-no-update/int-2.0BinaryTree.csv};

        \addplot+[dashed, red, mark=+] table [x=ttl, y=c95Mean, col sep=comma] {data/internal-time/init30/int-adv-2.0-c-1-30-no-update/int-2.0BinaryTree.csv};

        \addplot+[blue, mark=square] table [x=ttl, y=c95Mean, col sep=comma] {data/internal-time/init30/int-adv-2.0-c-0-30-no-update/int-2.0TernaryTree.csv};

        \addplot+[dashed, blue, mark=square] table [x=ttl, y=c95Mean, col sep=comma] {data/internal-time/init30/int-adv-2.0-c-1-30-no-update/int-2.0TernaryTree.csv};
    \end{axis}
    \end{tikzpicture}
    \caption{$\lambda_\intl=\lambda_\max$}
    \label{fig:95-correct-1}  
\end{subfigure}
\hfill
\begin{subfigure}{0.47\linewidth}
    \begin{tikzpicture}
        \begin{axis}[
        % legend pos=north west,
        % ymin=200,
        % ymax=300,
        % xmin=480,
        % xmax=600,
        % legend columns=2,
        xlabel={$\ttl$},
        grid=minor,
        % mark repeat=10,
        % mark phase=10,
        % mark size=1.5pt,
        line width=0.5,
        legend columns=6,
        legend entries={U-$C_0$, U-$C_0$, B-$C_0$, B-$C_1$, T-$C_0$, T-$C_1$},
        legend to name=95-correct,
        mark options=solid,
        ] 
        \addplot+[black, mark=x] table [x=ttl, y=c95Mean, col sep=comma] {data/internal-time/init30/int-adv-1.0-c-0-30-no-update/int-1.0UniRandom.csv};

        \addplot+[dashed, black, mark=x] table [x=ttl, y=c95Mean, col sep=comma] {data/internal-time/init30/int-adv-1.0-c-1-30-no-update/int-1.0UniRandom.csv};

        \addplot+[red, mark=+] table [x=ttl, y=c90Mean, col sep=comma] {data/internal-time/init30/int-adv-1.0-c-0-30-no-update/int-1.0BinaryTree.csv};

        \addplot+[dashed, red, mark=+] table [x=ttl, y=c95Mean, col sep=comma] {data/internal-time/init30/int-adv-1.0-c-1-30-no-update/int-1.0BinaryTree.csv};

        \addplot+[blue, mark=square] table [x=ttl, y=c90Mean, col sep=comma] {data/internal-time/init30/int-adv-1.0-c-0-30-no-update/int-1.0TernaryTree.csv};

        \addplot+[dashed, blue, mark=square] table [x=ttl, y=c95Mean, col sep=comma] {data/internal-time/init30/int-adv-1.0-c-1-30-no-update/int-1.0TernaryTree.csv};

    \end{axis}
    \end{tikzpicture}
    \caption{$\lambda_\intl=2\lambda_\max$}
    \label{fig:95-correct-2}
\end{subfigure}
\caption{Time when 95\% of the devices in the network becomes 
correct starting from a initial fraction of 30\% corrupt devices in 
Binary~(B), \adhoc~(U), and Ternary~(T) topologies for varying $\ttl$. 
Solid and dashed lines corresponds to $C_0$ and $C_1$ respectively.}
\label{fig:95-correct}
\end{figure}
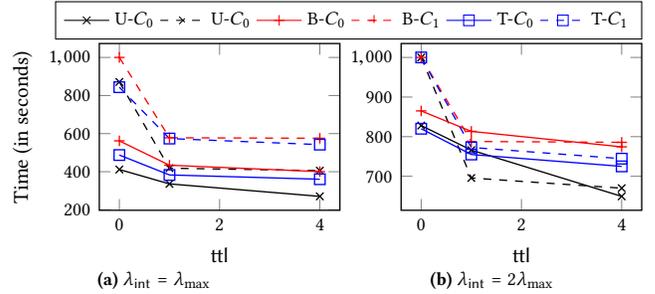

\noindent{\bf Update with different Configuration.}
Figure~\ref{fig:eval-update} illustrates the fraction of updated device 
over time in all topologies for configuration $C_0$ and  $C_1$ with 
$f=0.30$. In all the experiments, eventually, almost all
devices get updated. In both $C_0$ and $C_1$, update in tree 
topologies takes longer because, a single corrupt device can 
temporarily stop updates in its entire subtree. 
Update in tree topology for $C_0$ takes longer time than 
$C_1$, because corrupt devices are more evenly spread across 
the network and hence code update temporarily halts more often 
in $C_0$. Ternary tree topology has a faster update than 
Binary tree due shorter tree height and large average 
number of neighbors.
\begin{figure}[h!]
\begin{centering}
\ \ \ \ \ \ \ \ \ \  \ \ref{update}
\begin{subfigure}{0.47\linewidth}
    \centering
    \begin{tikzpicture}
        \begin{axis}[
        % legend pos=north west,
        ymin=0.5,
        xmin=480,
        xlabel={time (in seconds)},
        ylabel= {Fraction of devices)},
        grid=minor,
        mark repeat=10,
        mark phase=10,
        mark size=1.5pt,
        line width=0.5,
        legend columns=6,
        legend entries={B0;,B1;,U0;,U1;,T0;,T1},
        legend to name=update,
        mark options=solid,
        ]
        \addplot+[black, mark=x] 
        table [x=time, y=UnUpdateMean, col sep=comma] 
        {data/internal-formatted/init30/int-adv-1.0-c-0-30/int-1.0UniRandom0.csv};
        \addplot+[dashed, black, mark=x]
        table [x=time, y=UnUpdateMean, col sep=comma] 
        {data/internal-formatted/init30/int-adv-1.0-c-0-30/int-1.0UniRandom1.csv};

        \addplot+[red, mark=+] 
        table [x=time, y=UnUpdateMean, col sep=comma] 
        {data/internal-formatted/init30/int-adv-1.0-c-0-30/int-1.0BinaryTree0.csv};
        \addplot+[dashed, red, mark=+] 
        table [x=time, y=UnUpdateMean, col sep=comma] 
        {data/internal-formatted/init30/int-adv-1.0-c-0-30/int-1.0BinaryTree1.csv};

        \addplot+[blue, mark=square] 
        table [x=time, y=UnUpdateMean, col sep=comma] 
        {data/internal-formatted/init30/int-adv-1.0-c-0-30/int-1.0TernaryTree0.csv};
        \addplot+[dashed, blue, mark=square] 
        table [x=time, y=UnUpdateMean, col sep=comma] 
        {data/internal-formatted/init30/int-adv-1.0-c-0-30/int-1.0TernaryTree1.csv};
    \end{axis}
    \end{tikzpicture}
    \caption{Configuration $C_0$}
    \label{fig:update-c0}
\end{subfigure}
\hfill
\begin{subfigure}{0.48\linewidth}
    \centering
    \begin{tikzpicture}
    \begin{axis}[
        % legend pos=north west,
        ymin=0.5,
        xmin=480,
        xlabel={time (in seconds)},
        % ylabel= {Time (in milliseconds)},
        % ytick = {\empty},
        % extra y ticks={0,0.1,0.2,0.3},
        % extra y tick labels={0,0.1,0.2,0.3},
        grid=minor,
        mark repeat=10,
        mark phase=10,
        mark size=1.5pt,
        line width=0.5,
        mark options=solid,
        ] 
        \addplot+[blue, mark=x,] 
        table [x=time, y=UnUpdateMean, col sep=comma] 
        {data/internal-formatted/init30/int-adv-1.0-c-1-30/int-1.0BinaryTree0.csv};
        \addplot+[dashed, blue, mark=x,]
        table [x=time, y=UnUpdateMean, col sep=comma] 
        {data/internal-formatted/init30/int-adv-1.0-c-1-30/int-1.0BinaryTree1.csv};
        \addplot+[red, mark=square,] 
        table [x=time, y=UnUpdateMean, col sep=comma] 
        {data/internal-formatted/init30/int-adv-1.0-c-1-30/int-1.0UniRandom0.csv};
        \addplot+[dashed, red, mark=square,] 
        table [x=time, y=UnUpdateMean, col sep=comma] 
        {data/internal-formatted/init30/int-adv-1.0-c-1-30/int-1.0UniRandom1.csv};
        \addplot+[black, mark=+,] 
        table [x=time, y=UnUpdateMean, col sep=comma] 
        {data/internal-formatted/init30/int-adv-1.0-c-1-30/int-1.0TernaryTree0.csv};
        \addplot+[dashed, black, mark=+,] 
        table [x=time, y=UnUpdateMean,  col sep=comma] 
        {data/internal-formatted/init30/int-adv-1.0-c-1-30/int-1.0TernaryTree1.csv};
    \end{axis}
    \end{tikzpicture}
    \caption{Configuration. $C_1$}
    \label{fig:update-c1}  
\end{subfigure}
\caption{Fraction of updated devices in the presence of an internal adversary 
with $f=0.30$, malware spread rate $\lambda_\int=\lambda_\max$ for Binary~(B)
tree, Ternary~(T) tree, and \adhoc~(U) network topology. Figure~(a) and (b)
corresponds to $C_0$ and $C_1$ respectively. Solid and dashed lines correspond to 
$\ttl=0$ and $\ttl=1$ respectively.}
\label{fig:eval-update}
\end{centering}
\end{figure}
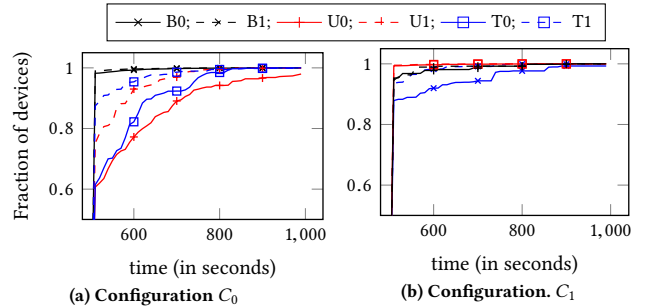

\subsection{External Adversary}
\label{sub:eval-external}
\pgfplotsset{small,label style={font=\fontsize{8}{9}\selectfont},legend style={font=\fontsize{7}{8}\selectfont},height=3.8cm,width=1.2\textwidth}
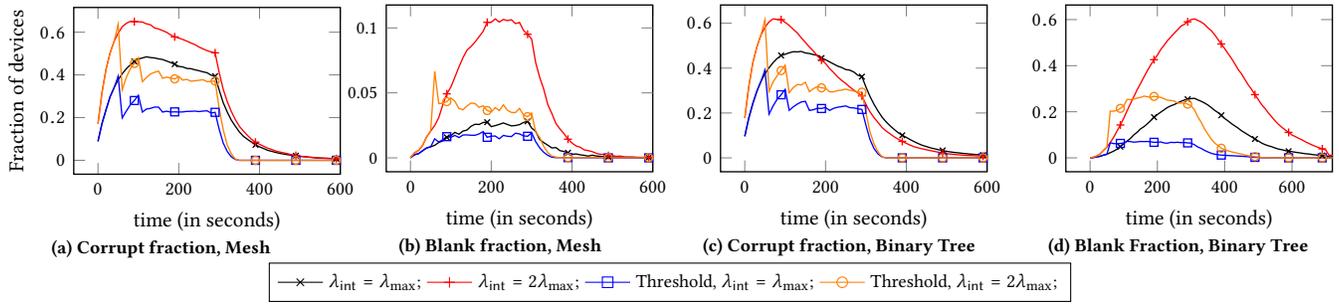
\begin{figure*}[th!]
\centering
\begin{subfigure}{0.24\linewidth}
\centering
    \begin{tikzpicture}
        \begin{axis}[
        % legend pos=north west,
        % ymin=0.5,
        % ymax=300,
        % xmin=480,
        xmax=600,
        xlabel={time (in seconds)},
        ylabel= {Fraction of devices},
        grid=minor,
        mark repeat=10,
        mark phase=10,
        mark size=1.5pt,
        line width=0.5,
        legend columns=8,
        legend entries={
        {$\lambda_\intl =\lambda_\max$;}, {$\lambda_\intl=2\lambda_\max$;},{Threshold, $\lambda_\intl =\lambda_\max$;}, {Threshold, $\lambda_\intl=2\lambda_\max$;}},
        legend to name=external,
        mark options=solid,
        ] 
        \addplot+[black, mark=x] table [x=time, y=UnDetectMean, col sep=comma] {data/external/ext-adv-1.0/ext-1.0UniRandom1.csv};

        \addplot+[red, mark=+] table [x=time, y=UnDetectMean, col sep=comma] {data/external/ext-adv-2.0/ext-2.0UniRandom1.csv};

        \addplot+[blue, mark=square] table [x=time, y=UnDetectMean, col sep=comma] {data/external/ext-adv-1.0/ext-1.0UniRandomTh1.csv};

        \addplot+[orange, mark=o] table [x=time, y=UnDetectMean, col sep=comma] {data/external/ext-adv-2.0/ext-2.0UniRandomTh1.csv};
    \end{axis}
    \end{tikzpicture}
    \caption{Corrupt fraction, \adhoc}
    \label{fig:ext-uni-corrupt}
\end{subfigure}
\hfill
\begin{subfigure}{0.24\linewidth}
    \centering
    \begin{tikzpicture}
    \begin{axis}[
        % legend pos=north west,
        % ymin=0.5,
        % ymax=0.33,
        % xmin=480,
        xmax=600,
        xlabel={time (in seconds)},
        ytick = {\empty},
        extra y ticks={0,0.05,0.1},
        extra y tick labels={0,0.05,0.1},
        grid=minor,
        mark repeat=10,
        mark phase=10,
        mark size=1.5pt,
        line width=0.5,
        mark options=solid,
        ] 
        \addplot+[black, mark=x] table [x=time, y=UnCorrectMean, col sep=comma] {data/external/ext-adv-1.0/ext-1.0UniRandom1.csv};

        \addplot+[red, mark=+] table [x=time, y=UnCorrectMean, col sep=comma] {data/external/ext-adv-2.0/ext-2.0UniRandom1.csv};

        \addplot+[blue, mark=square] table [x=time, y=UnCorrectMean, col sep=comma] {data/external/ext-adv-1.0/ext-1.0UniRandomTh1.csv};

        \addplot+[orange, mark=o] table [x=time, y=UnCorrectMean, col sep=comma] {data/external/ext-adv-2.0/ext-2.0UniRandomTh1.csv};
    \end{axis}
    \end{tikzpicture}
    \caption{Blank fraction, \adhoc}
    \label{fig:ext-uni-blank}  
\end{subfigure}
\hfill
\begin{subfigure}{0.24\linewidth}
\centering
    \begin{tikzpicture}
        \begin{axis}[
        % legend pos=north west,
        % ymin=0.5,
        % ymax=300,
        % xmin=480,
        xmax=600,
        xlabel={time (in seconds)},
        grid=minor,
        mark repeat=10,
        mark phase=10,
        mark size=1.5pt,
        line width=0.5,
        legend columns=6,
        mark options=solid,
        ] 

        \addplot+[black, mark=x] table [x=time, y=UnDetectMean, col sep=comma] {data/external/ext-adv-1.0/ext-1.0BinaryTree1.csv};

        \addplot+[red, mark=+] table [x=time, y=UnDetectMean, col sep=comma] {data/external/ext-adv-2.0/ext-2.0BinaryTree1.csv};

        \addplot+[blue, mark=square] table [x=time, y=UnDetectMean, col sep=comma] {data/external/ext-adv-1.0/ext-1.0BinaryTreeTh1.csv};

        \addplot+[orange, mark=o] table [x=time, y=UnDetectMean, col sep=comma] {data/external/ext-adv-2.0/ext-2.0BinaryTreeTh1.csv};
    \end{axis}
    \end{tikzpicture}
    \caption{Corrupt fraction, Binary Tree}
    \label{fig:ext-bin-corrupt}
\end{subfigure}
\hfill
\begin{subfigure}{0.24\linewidth}
    \centering
    \begin{tikzpicture}
    \begin{axis}[
        % legend pos=north west,
        % ymin=0.5,
        % ymax=0.33,
        % xmin=480,
        xmax=720,
        xlabel={time (in seconds)},
        % ylabel= {Time (in milliseconds)},
        % ytick = {\empty},
        % extra y ticks={0,0.1,0.2,0.3},
        % extra y tick labels={0,0.1,0.2,0.3},
        grid=minor,
        mark repeat=10,
        mark phase=10,
        mark size=1.5pt,
        line width=0.5,
        mark options=solid,
        ] 
        \addplot+[black, mark=x] table [x=time, y=UnCorrectMean, col sep=comma] {data/external/ext-adv-1.0/ext-1.0BinaryTree1.csv};

        \addplot+[red, mark=+] table [x=time, y=UnCorrectMean, col sep=comma] {data/external/ext-adv-2.0/ext-2.0BinaryTree1.csv};

        \addplot+[blue, mark=square] table [x=time, y=UnCorrectMean, col sep=comma] {data/external/ext-adv-1.0/ext-1.0BinaryTreeTh1.csv};

        \addplot+[orange, mark=o] table [x=time, y=UnCorrectMean, col sep=comma] {data/external/ext-adv-2.0/ext-2.0BinaryTreeTh1.csv};
    \end{axis}
    \end{tikzpicture}
    \caption{Blank Fraction, Binary Tree}
    \label{fig:ext-bin-blank}  
\end{subfigure}
\ref{external}
\caption{Fraction of corrupt and \blank\ device in \adhoc\ and Binary Tree topology in the presence of an external adversary \adv$_\extl$. Solid lines
corresponds to the adaptive case with $\ttl=1$ and dashed lines 
corresponds to non-adaptive case, i.e., $\ttl=0$. Plots for the situation
where the interarrival time between two consecutive self-checks are }
\label{fig:external}
\end{figure*}

\vspace{0.5mm}
\noindent{\bf Varying network topology.}
Red and black plots in Figure~\ref{fig:ext-uni-corrupt} 
and~\ref{fig:ext-bin-corrupt} illustrates the fraction 
of corrupt devices for \adhoc\ and Binary tree topology 
with $\lambda_\extl=\lambda_\max, 2\lambda_\max$ respectively. 
We omit the results for the Ternary tree as it is very similar to the 
results of Binary tree topology. 
For both topologies, the fraction of undetected corrupt device increases 
approximately until 100 seconds and then gradually starts decreasing. This 
is because we initialize $\lambda =\lambda_\max=1/100$.
Interestingly, after 100 seconds, although the adversary is corrupting 
additional devices, the fraction of corrupt nodes decreases even for
$\lambda_\extl=2\lambda_\max$. 
This is because adversary randomly chooses device for corruption and 
since more than $50\%$ of the devices are already corrupt or blank by 
time 100s, the effective corruption rate is lower than $\lambda_\max$. 

The fraction of blank devices is higher in Binary tree
topology because blank devices in Binary tree have fewer 
neighbors, and hence they remain blank till one of their neighbor corrects 
itself. Alternatively, in \adhoc\ topology average number of neighbors
per device is higher, which increases the likelihood of one honest
neighbor at any given time. 
This also explains the rapid decrease in fraction of  corrupt device 
in Binary tree topology. As blank devices are immune to
corruption, the probability of corrupting an honest device is lower 
in Binary tree topology. Once \adv$_\extl$ is disconnected 
from the network, i.e., after 300s, the rate of reduction in 
fraction of corrupt device follows the tail of exponential
distribution.

\vspace{0.5mm}
\noindent{\bf Varying Spread Rate.}
Black plot in Figure~\ref{fig:external} corresponds to
$\lambda_\extl=\lambda_\max$ and red plot corresponds to
$\lambda_\extl=2\lambda_\max$. As expected
with higher $\lambda_\extl$, a higher fraction of devices gets 
corrupt. Also, for the same reason, the fraction of the correct device
is lower for the lower corruption rate. Interestingly, the fraction 
of corrupt devices is higher in  Binary tree topology after 200s for
$\lambda_\extl = 2\lambda_\max$. As discussed earlier, this is
due to the higher fraction of blank devices in the network which 
reduces the effective malware spread rate.

\vspace{0.5mm}
\noindent{\bf Thresholding self-check rate.}
So far we have only considered the scheme where the interarrival
time between consecutive self-checks at honest devices are 
drawn from an exponential distribution with parameter $\lambda$.
An issue with this approach is the unbounded interarrival time. 
Thus, we evaluate \dante\ with the modification 
where we upper bound the self-check period by 50 seconds, i.e.,
$\delta \gets \min\{\exp\{\lambda\},50\}$. 
Blue and orange plots in Figure~\ref{fig:external} illustrates 
this results in the presence of \adv$_\extl$. Observe that,
thresholding reduces the fraction of corrupt and blank devices.
This is because, devices are performing frequent self-checks, 
hence detecting the malware earlier. But this comes at the cost
of higher energy usage. 
Also, for both \adhoc\ and binary tree topologies, at time 
instants that are multiples of 50, a large number of devices 
detects and hence corrects themselves. Interestingly, the fraction 
of corrupt device does not reach zero at time instants because 
additional devices whose first self-check interval was less 
than 50 gets corrupted.

\subsection{Performance}
\label{sub:eval performance}
As we describe in~\S\ref{sub:memory cost}, our bloom filter
based approach requires $\ell|L|+\mu Z/t$ bits of \sram\ 
space in contrast to $\ell Z/t$ bits of space using the 
naive approach. 
Therefore, for $Z=16384$ Bytes, i.e., 16 KB which is typically the case 
with Texas Instrument's MSP430 micro-controllers~\cite{msp2018dang}, 
for $t=256$ Bytes, $|L|=4$ and $\mu=8$ we will only require 128 
Bytes of additional space in \sram. This gives us an $8\times$
improvement over the naive system with $\ell=128$. Moreover, by 
substituting these numbers in equation~(\ref{eq:expected chunk}),
we get that the expected number of chunks, a blank device needs
to download for $\kappa=4$, i.e., when adversary modifies content 
of four chunks, is $\approx10$. This is $6\times$ better than the
naive scheme of downloading all the chunks. 
Lastly, using equation~(\ref{eq:expected number}), we get that
the expected number of the honest neighbors who will transmit 
the requested chunks for different values of $m$, the number of 
honest neighbor in the  worst case are:
\begin{table}[h!]
  \begin{tabular}{cccccc} 
      \hline
      \# neighbors, m & 2 & 5 & 10 & 20 \\ 
      \hline
      E[\# neighbors to transmit code] & 1.50 & 1.57 & 1.57 & 1.58 \\ 
      \hline
  \end{tabular}
  \label{tab:permmission}
\end{table}

\noindent
This shows even in dense network, in expectation, less than two
neighbors will end up transmitting the requested chunks. This 
is significantly better than all the naive approaches.
% !TEX root = ../main.tex
\section{Related Work}
\label{sec:related work}
%\vinay[We have to make clear how our work is different. We are doing something much beyond attestation. So it is not clear why we need to mention the following works
%  on attestation.  We can maybe say that our work falls into the genre of Device Swarm security. While our work focuses on correction and  updation, most of the previous work
 % looked only at attestation. Then briefly mention the work on attestation. Then discuss any work which has done correction and/or updation and contrast it with our work. If we
%are borrowing ideas from some other work, then we should mention those clearly and say how our work is similar.]

\dante falls into the genre of Device Swarm Security.
%The proposed approach requires secure storage, to hold the read-only attestation code and attestation related keys, as well as an environment where tasks are executed in an uninterrupted and atomic manner.
While our work focuses on correction and  updation in the presence 
of an adversary, most of the previous works looked only at only
attestation in the presence of an adversary or updation with no adversary.

%Device attestation techniques can be classified in to two as  \textit{Single Prover} \cite{li2010sbap},
%\cite{seshadri2006scuba} or \textit{Swarm Attestation} \cite{asokan2015seda}, \cite{carpent2017lightweight}. 
%Devices in Industrial IoTs often work as a group and are mobile in nature \cite{csahin2004swarm}. In a single
%prover approach, the verifier attests each target device one at a time. Such an approach lacks flexibility and does not
%scale to systems that work in a group.

Proposals such as~\cite{asokan2015seda, carpent2017lightweight, ambrosin2016sana} assume a {\em Single External Verifier}
to carry out swarm attestation while others~\cite{ibrahim2018aid, wedaj2019dads} use a \textit{Decentralized} approach where each
member device is attested by a genuine node in its neighborhood.
Ambrosin M. et al. designed a collective attestation scheme for IoT
swarms for Highly Dynamic Swarm Topologies~\cite{ambrosin2018pads}.
% These techniques typically assume minimal hardware security support
% (i.e., integrity of the measurement entity and secure storage), as
% provided by SMART \cite{eldefrawy2012smart}, TrustLite
% \cite{koeberl2014trustlite} and TyTAN~\cite{brasser2015tytan}
% architectures. 
These methods do not address updating or correction of code,  however.

In SAFE$^{d}$~\cite{visintin2019safe}, a pair of embedded devices
in a swarm attest to each other without the need of an external 
verifier. Similar to~\cite{wedaj2019dads}, SAFE$^{d}$ also
removes a single-point-of-failure issue by allowing swarm 
members to coordinate and self-protect the underlying network. 
The SAFE$^{d}$ network forms multiple overlays among swarm 
members that replicate proofs indicating the correctness of 
prover devices.  
Recently, Ibrahim et al. proposed HEALED~\cite{ibrahim2019healed},
a new attestation scheme capable of detecting corrupt device
and healing upon compromise. Every corrupt device in HEALED,
interact with honest device to localize modified memory regions.
This approach requires $O(\log z)$ rounds of communication to 
identify a single corrupt region for a application program of 
size $z$. Also, none of the above approaches consider propagating
malware. Furthermore, they make strong assumptions such as: 
a corrupt device voluntarily tries to correct itself; every 
device in the network can perform securely route messages
to other honest devices despite the presence of a Byzantine 
adversary in the network. 

Regarding code updation, N. Asokan et al. extended The Update
Framework~\cite{samuel2010survivable} and proposed an
architecture for secure firmware update~\cite{asokan2018assured}.
This work takes various stakeholders such as manufacturer, 
software distributor, domain controller and end devices in 
IoT firmware update ecosystem; and establishes an end-to-end 
security between devices manufactures and IoT devices. 
This work also suffers from the limitations described 
in~\S\ref{sec:code update}.

\section{Conclusion}
\label{sec:conclusion}
In this paper, we presented \dante\ - a novel decentralized, scalable, efficient, and secure mechanism of recovering a network of heterogeneous low-end devices in the presence of self-propagating malware. 
%\dante\ tackles both internal and external attacks. 
Furthermore, unlike prior works, \dante\ guarantees update 
of entire network. For efficiency, we used bloom-filters to identify compromised code chunks, random back-off and stream signatures to 
reduce bandwidth overhead and enhance security. Evaluation, of 
our approach using OMNeT++, illustrates that \dante\ scales upto 1000s
of device and can recover the entire network in minutes. 
We also evaluated the memory and communication costs of \dante\ and 
showed that it incurs very low overhead.
Addressing these issues with dynamic swarms and run-time 
attacks could be an interesting avenue for future researches.

% We demonstrated \dante\ atop of OMNeT++ simulation environment based % state-of-the-art architecture for embedded devices, TrustLite~\cite{koeberl2014trustlite}, and perform a rigorous evaluation of it with various realistic network topologies. 
% We also evaluated the memory and communication costs of \dante, and showed that it is: (1) secure; (2) efficient (member devices recover/update quickly), and (3) incurs very low overhead. Furthermore, \dante\ ensures member device's state recovery without the intervention of an external trusted entity. 

% different choices could be an interesting avenue for future researches. We would like to extend \dante considering dynamic swarms and run-time attacks (i.e., where an attacker alters application's behavior without modifying the binary) so that we can have a more secure and robust IoT network with a minimal performance overhead. 

%\begin{itemize}
%    \item Possibility of improvements using trusted hardwar or TPMs in the device
%\end{itemize}

\begin{acks}
The authors would like to thank Nitin Awathare, Aashish Kolluri, 
Jong Chan Lee, Archit Patke, Soundarya Ramesh, Ling Ren, and Qi 
Wang for helpful discussion and feedback on the early version 
of the paper. 
\end{acks}

\bibliographystyle{ACM-Reference-Format}
\bibliography{main}

\appendix
\section{Algorithms}
\begin{algorithm}[h!]
  \caption{\selfcheck}\label{algo:selfcheck}
  \begin{algorithmic}[1]
        \State {\bf Input} $ak_i, v_i$, $[\code]$
        \State $v'_i \gets$ \attest$(ak_i,[\code])$
        \If {$v'_i = v_i$}
        \State $\lambda \gets \max\{\lambda/(\lambda+1),\lambda_\min\}$; \
        $\delta \gets \exp\{\lambda\}$
        \State {\bf enable} interrupt; {\bf resume} application.
      \Else
      \State {\bf mark} $[\code]$ as non-executable.
      \State \rectify$()$
      \EndIf
  \end{algorithmic}
\end{algorithm}

\begin{algorithm}[h!]
  \caption{${\sf handleCodeRequest}$ at device $n_j\in N_i$}
  \label{algo:handle-req}
  \begin{algorithmic}[1]
        \State {\bf input} $\msg_\req=\langle \ttl,s_i,|N_i|,z_i,b,\Pi\rangle$
        \State {\bf global} $A_j, \Delta, \theta, \lambda, \lambda_\max$
        \If {$b \in A_j$}
        \State $z_j \gets \ver(b)$
        \If {$z_j \ge z_i$}
        \State $\tau_j \gets  (\Delta -(z_j-z_i))|N_i|\theta + 
        \left\lfloor{\mathcal U}(0,1)|N_i|\right\rfloor\theta $
        \State {\bf set} transmit code timer after $\tau_j$ 
        \EndIf
        \EndIf
        \If {$\ttl>0$}
        \State $\lambda \gets \min \{ 2\lambda, \lambda_\max \}$; 
        \ $\ttl \gets \ttl-1$
        \EndIf
        \If {$\ttl>0$}
        \State {\bf broadcast} a message warning devices in $N_j$
      \EndIf
  \end{algorithmic}
\end{algorithm}

\begin{algorithm}[h!]
    \caption{${\sf handleCodeResponse}$ at device $n_i$}
    \label{algo:handle-resp}
    \begin{algorithmic}[1]
        \State {\bf input} $\Pi$
        \While {true}
            \State new response $\msg_\resp=\langle\{\posi_j,\dati_j\}\rangle$ 
            \While {{\bf next} $\posi \in \{\Pi\ \cap\ \msg_\resp\}$}
                \State $\dati \gets \msg_\resp[\posi]$
                \If {stream signature of $\dati$ is valid} 
                \State {\bf load} $\dati$ to \code; 
                \State $\Pi \gets \Pi\setminus\posi;
                \ \msg_\resp \gets \msg_\resp\setminus\{\posi,\dati\}$
                \Else
                \State {\bf break}
                \EndIf
            \EndWhile
            \If{$\Pi$ is {\bf empty}}
            \State $\lambda\gets\lambda_\max$; {\bf update} interrupt handler
            \State {\bf broadcast} $\msg_{\sf done}$; {\bf restart} $n_i$
            \Else
            \State $\delta \gets \exp\{\lambda\}$
            \State re-broadcast $\msg_\req$ after $\delta$ time interval.
            \EndIf
        \EndWhile 
    \end{algorithmic}
\end{algorithm}

\section{Notation Table}
\begin{table}[h!]
\centering
\begin{tabular}{|c|c|} 
    \hline
    Notation & Description \\ 
    \hline
    \opt & Network operator/Owner \\
    $pk_{\mathcal O}$, $sk_{\mathcal O}$ & public-private key pair of \opt \\
    $N$ & Total number of nodes in the network \\
    $n_i$ & $i^{th}$ device \\
    $pk_i, sk_i$ & public-private key pair of $n_i$ \\
    $B_i$, $C_i$ & Application binaries executed and stored by $n_i$ \\ 
    $N_i$ & Neighbors of $n_i$ after successful rendezvous \\
    $N^{(b)}_i$ & Devices in $N_i$ that runs/stores binary $b$ \\
    $k_i$ & Symmetric key shared by $n_i$ \\
    $q_i$ & Sequence number of $n_i$ \\
    $\lambda$ & Self-check rate \\
    $ak_i$, $v_i$ & Attestation key and value at $n_i$ \\
    $F$, $L=\{l_j\}$ & Bloom filter and the corresponding keys \\
    \adv$_\intl$, \adv$_\extl$ & Internal and External Adversary resp. \\
    $\lambda_\intl$, $\lambda_\extl$ &  Corruption rate of \adv$_\intl$, and \adv$_\extl$ \\
    \hline
\end{tabular}
\label{tab:notation}
\end{table}

\section{Proofs}
\label{sec:proofs}
To prove Theorem~\ref{thm:expected number} we will first
prove Lemma~\ref{lem:k-transmisson}. Let the time interval where 
devices running an identical version of the code, transmits the 
requested chunk be called as an {\em epoch}. Observe that, each 
epoch $m\theta$ long. Let each epoch be divided into $m$ time
intervals called a {\em slot}. Note that, in \dante\ devices
running different versions always sends the requested chunks 
in disjoint epochs. Also, within an epoch, once an honest device
sends the requested chunk, the recipient device broadcasts to 
each of its neighbors to stop them from redundantly sending 
the same chunks~(ref.~\S\ref{sub:correction}). Hence, if only
neighbor sends the requested chunk in the first non-empty 
slot, there would be no-redundancy at all. 
\begin{lemma}
Let there be $m$ neighbors running the same version of application 
$b$ for any given device. 
Also, let $X^j$ for $1\le j\le m-1$ be the random variable 
denoting the number of devices which sends the requested chunks 
in slot $j$ when the remaining $m-X^j$ nodes transmits in a 
slot greater than $j$. Then, 
\begin{equation}
\Pr[X^j=k] = \frac{\binom{m}{k} (m-j)^{(m-k)}}{m^m}
\end{equation}
\label{lem:k-transmisson}
\end{lemma}
\begin{proof}
We use counting arguments to prove this theorem. There are total 
$m^m$ possibilities of arranging $m$ devices in $m$ slot. Among
these possibilities, the number of ways $X_j=k$ can occur if:~(\romannumeral1) Any subset of $k$ devices transmits in slot 
$j$; there are $\binom{m}{k}$ possible ways of selecting 
these devices. and~(\romannumeral2) the remaining $m-k$ device 
transmits in slots $j+1$ to $m$; there are $(m-j)^{m-k}$ ways 
of doing this. Putting them together gives us the desired 
result.
\end{proof}

\begin{corollary}
The probability that only one neighbor device will transmit the
requested code in the situation described above is:
\begin{equation}
\Pr[\text{One device transmits code}] = \sum_{j=1}^{m-1}\left(1-\frac{j}{m}\right)^{(m-1)}
\end{equation}
\end{corollary}

\begin{proof} (Theorem~\ref{thm:expected number})
The first term directly follows Lemma~\ref{lem:k-transmisson} and the
second term corresponds to the case where all device transmits the 
requested chunk in the $m^{\rm th}$ slot.
\end{proof}

\begin{proof}(Theorem~\ref{thm:recover})
Let $\mathbf{H}_b$ denote the set of honest devices at any given time 
after $t_0$, by our assumption $|\mathbf{H}_b|>0$. Consider a corrupt or 
blank device $n_d\in G_b$ which is initially $\nu$ hops away from the 
nearest honest device in $n_h \in \mathbf{H}_b$ and $n_a$ be the 
penultimate node on the path from $n_d$ to $n_h$. By definition, 
$n_a$ is corrupt. After an exponentially distributed waiting time, 
$n_a$ performs a self-check and recovers itself one of its honest
neighbors. This is guaranteed to happen since $n_h \in N^(b)_a$. 
Thus $d_i$ now joins $\mathbf{H}_b$. This reduces the distance
between $n_d$ and the nearest honest by at least one unit. Since,
there are only finite number of nodes in the network, this 
distance will eventually become zero. 
\end{proof}

\begin{proof} (Theorem~\ref{thm:update})
Let $U_b\subseteq G_b$ be the set of updated devices at time $t_0$. 
By definition all non-updated devices in the neighbourhood of devices
in $U_b$. Thus, whenever one of these devices say $n_v$ performs
self-check, it will get updated with the newer version and join 
$U_b$. Device $n_v$ then broadcast the newer version to update all
non-updated honest devices in its neighborhood. Hence, the number
of updated device increases by at least one. Since, there are only
finite number of devices, eventually the entire network will get
updated.
\end{proof}

\end{document}